\documentclass[journal, twocolumn]{IEEEtran}
\ifCLASSINFOpdf
   \usepackage[pdftex]{graphicx}
   \graphicspath{{../pdf/}{../jpeg/}}
   \DeclareGraphicsExtensions{.pdf,.jpeg,.png}
\else
   \usepackage[dvips]{graphicx}
   \graphicspath{{../eps/}}
\fi

\ifCLASSOPTIONcompsoc
  \usepackage[caption=false,font=normalsize,labelfont=sf,textfont=sf]{subfig}
\else
  \usepackage[caption=false,font=footnotesize]{subfig}
\fi

\usepackage[cmex10]{amsmath}
\usepackage{bm}
\usepackage{amssymb}
\usepackage{amsthm}

\usepackage[ruled,linesnumbered]{algorithm2e}
\usepackage{algpseudocode}
\usepackage{array}
\usepackage{cite}
\usepackage{url}
\usepackage{booktabs}
\usepackage{ragged2e}%两端对齐
\usepackage{color}
\usepackage{eso-pic}

\renewcommand{\raggedright}{\leftskip=0pt \rightskip=0pt plus 0cm}%两端对齐

\renewcommand{\raggedright}{\leftskip=0pt \rightskip=0pt plus 0cm}%两端对齐，需要的时候用\raggedright 即可

\newcommand{\etal}{\emph{et~al.~}}

\newtheorem{lemm}{Lemma}[section]
\newtheorem{theo}{Theorem}[section]

\makeatletter

\newcommand{\Rmnum}[1]{\expandafter\@slowromancap\romannumeral #1@}
\makeatother

\begin{document}

\title{Joint Transportation and Charging Scheduling in Public Vehicle Systems - A Game Theoretic Approach}

\author{Ming~Zhu,
        ~Xiao-Yang~Liu,
        ~and~Xiaodong~Wang,~\IEEEmembership{Fellow,~IEEE}
        % <-this % stops a space
%Ming~Zhu,~\IEEEmembership{Student Member,~IEEE}
%\thanks{E-mail: \{zhumingpassional\}@gmail.com}
\thanks{M.~Zhu is with the Shenzhen Institutes of Advanced Technology, Chinese Academy of Sciences, Shenzhen, China, and the Department of Computer Science and Engineering, Shanghai Jiao Tong University, Shanghai, China, E-mail: zhumingpassional@gmail.com, zhumingpassional@sjtu.edu.cn.}
\thanks{X.-Y.~Liu is with the Electrical Engineering Department, Columbia University, New York City, US, and the Department of Computer Science and Engineering, Shanghai Jiao Tong University, Shanghai, China. E-mail: xiaoyang@ee.columbia.edu.}% <-this % stops a space
\thanks{X.~Wang is with the Department of Electrical Engineering, Columbia University, New York, NY, 10027, USA (E-mail: wangx@ee.columbia.edu).}
}
%\markboth{IEEE TRANSACTIONS ON INTELLIGENT TRANSPORTATION SYSTEMS,~Vol.~xx, No.~xx, May~2017}%
%{Shell \MakeLowercase{\textit{et al.}}: Bare Demo of IEEEtran.cls for Journals}

\maketitle
%%% Group authors per affiliation:

\begin{abstract}
  Public vehicle (PV) systems are promising transportation systems for future smart cities which provide dynamic ride-sharing services according to passengers' requests. PVs are driverless/self-driving electric vehicles which require frequent recharging from smart grids. For such systems, the challenge lies in both the efficient scheduling scheme to satisfy transportation demands with service guarantee and the cost-effective charging strategy under the real-time electricity pricing. In this paper, we study the joint transportation and charging scheduling for PV systems to balance the transportation and charging demands, ensuring the long-term operation. We adopt a cake cutting game model to capture the interactions among PV groups, the cloud and smart grids. The cloud announces strategies to coordinate the allocation of transportation and energy resources among PV groups. All the PV groups try to maximize their joint transportation and charging utilities. We propose an algorithm to obtain the unique normalized Nash equilibrium point for this problem. Simulations are performed to confirm the effects of our scheme under the real taxi and power grid data sets of New York City. Our results show that our scheme achieves almost the same transportation performance compared with a heuristic scheme, namely, transportation with greedy charging; however, the average energy price of the proposed scheme is 10.86\% lower than the latter one.
\end{abstract}


\IEEEpeerreviewmaketitle

\section{Introduction}\label{Sec:Introduction}

%Shared connected vehicles are among the most important ``things" in the Internet of Things revolution \cite{amadeo2016connected-vehicles}. It has been predicted that, by 2020, 50 billion devices will be connected to the Internet \cite{vaquero2014-50billion-connected-devices-2020}, including 250 million vehicles \cite{haberle2015connected-car}, accounting for about 0.5\%. Internet of Things are driving the evolution of conventional Vehiclular Ad-hoc Networks (VANETs) into the Internet of Vehicles (IoVs) \cite{wang2016SoftwareDefined-IoV}, which employ technologies such as vehicle-to-vehicle communications, vehicle-to-infrastructure communications  \cite{viriyasitavat2015vehicular-communication}, vehicle-to-human communications, collision avoidance, traffic control \cite{zhu2013lane-change}, and driverless/self-driving cars. A conventional VANET turns vehicles into moving wireless access points to exchange messages among vehicles, providing wireless connectivity to other vehicles (and humans) in their vicinity. In contrast, IoV expands VANET by turning every vehicle into a smart object equipped with multiple powerful sensors, communication units, computing units, and data/energy storage devices.

The public vehicle (PV) systems \cite{zhu2016PublicVehicle} \cite{zhu2016transfer} \cite{Zhu2016TrafficBigData}, also known as the shared internet of vehicle systems or intelligent transportation systems, provide high-quality ride-sharing services in future smart cities. The vehicles in PV systems, called PVs, are typically driverless/self-driving \cite{tao2017driverless-car} electric vehicles with large capacities just like buses. PVs are connected to smart grids for self-charging. A PV system consists of three main components: a cloud, passengers/users, and PVs.
\begin{figure}
  \centering
  \includegraphics[height=0.60\linewidth,width=0.98\linewidth]{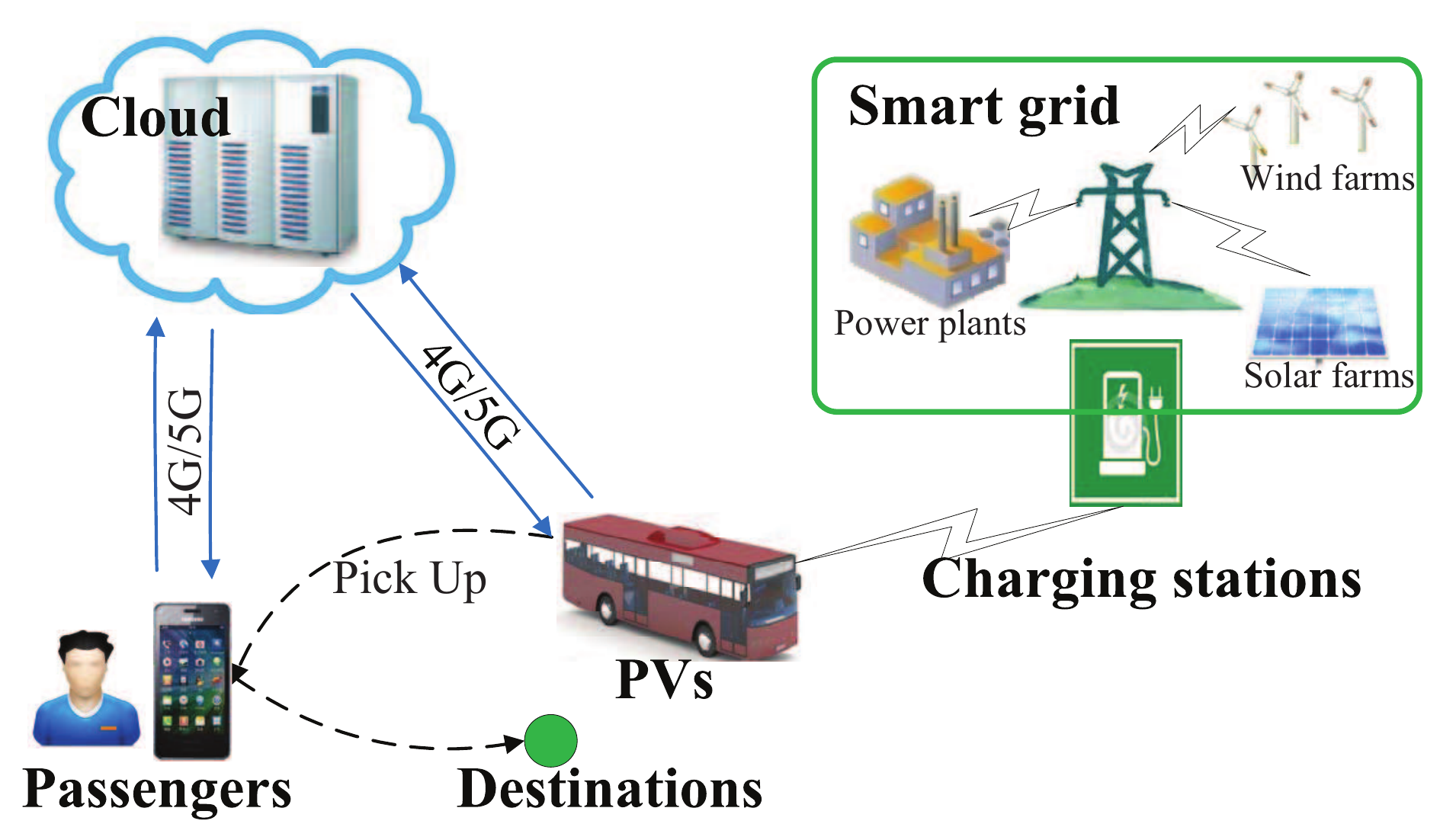}
  \caption{The V2G scenario: a PV system and the smart grid.}
  \label{Fig:Scenario}
\end{figure}
The operation flow of a PV system is as follows. If a passenger/user needs a trip service, he/she sends a request to the cloud via a smart phone, including an earliest start time, a pickup position (origin) and a dropoff position (destination), etc. Then the cloud computes the ride matches between PVs and passengers, and calculates paths for PVs, and finally schedules a suitable PV to drive him/her from the origin to the destination, wherein the paths may be shared with others.

PVs are typically electric vehicles (EVs), and are connected to smart grids for battery charging. Smart grids are envisioned as the next-generation power grid systems that can intelligently accommodate requests by all connected users. They are equipped with a smart metering infrastructure capable of sensing and measuring power consumptions from consumers with the integration of advanced computing, control, information and communication technologies \cite{yan2013survey-smart-grid}. The smart grids will have more efficient, more economical and more reliable power generations, distributions, and consumptions than the conventional power grids. PVs and smart grids constitute complicated vehicle-to-grid (V2G) ecosystems \cite{tan2016gameV2G}. Fig.~\ref{Fig:Scenario} shows an overview of the target scenario including a PV system and a smart grid.

To the best of our knowledge, this is the first work to consider the joint transportation and charging scheduling problem for PV systems. Existing works focus only on either the transportation or the charging problem for PV systems or PV-like systems. For example, \cite{santi2014VehiclePooling} proposes path planning strategies without considering energy storage or charging actions; whereas \cite{chen2016Traffic-Demand-Forecast} proposes the charging strategies from smart grids neglecting the transportation actions. In addition, the charging scenario in PV systems is different from that in existing works. For example, private EVs are parked 95\% of the time \cite{morency2015-CarParkingTime-95-percent} and the batteries are usually charged at parking lots at homes or working places most of the day. PVs have much less parking time than private EVs since they are shared by all passengers, and they need frequent charging since they serve more passengers with high occupancy rates \cite{zhu2016PublicVehicle}. PVs' routes and schedules vary with times and locations for the constantly changing transportation demands. Therefore, how to take advantage of the limited amount of idle time to charge the batteries becomes an important problem. We consider the joint transportation and charging scheduling problem (JTCSP) for PV systems. The goal is to balance the transportation and charging demands to guarantee the long-term operation of PV systems with less charging costs and more profits.

The challenges of JTCSP are as follows: 1) The cloud should consider the transportation demands from passengers which vary with times and locations. 2) The charging demands of PV groups should be satisfied with less costs under the scenario of real-time electricity pricing. 3) The cloud should ensure PVs to have sufficient energy to serve passengers at any time even if the energy price is high, otherwise, some PVs may not have sufficient energy to serve passengers. Too many or too few PVs engaged in transportation or charging will affect the profits and energy storages, e.g., some PVs engaged in transportation may have insufficient energy, and some other PVs engaged in charging may not provide transportation services to earn profits.

The main contributions of this paper include the following: 1) We adopt the cake cutting game \cite{chen2010cake-cutting} to capture the transportation and charging patterns of PVs, and then use a novel utility model to reflect the transportation benefits, satisfactions for charged batteries, and charging costs. Suppose each player has a different opinion as to which part of a cake is more valuable. The aim of a cake cutting game \cite{chen2010cake-cutting} is to divide a cake among multiple players so that everyone values his or her piece no less than any other piece. In this paper, we make use of the cake cutting game in selecting different transportation and charging vehicles for different PV groups. 2) We analyze the existence and uniqueness of the Nash equilibrium in this cake cutting game, and analyze the features of the game in the JTCSP compared with existing works, e.g., those that only consider transportation or charging. 3) We propose an efficient solution to JTCSP to balance the transportation and charging demands of PVs which achieves a unique normalized Nash equilibrium point in the cake cutting game. 4) Simulations are performed to demonstrate the effectiveness of the proposed solution under the real transportation and power grid data sets of New York City, and then compare the performance of our scheme with a heuristic scheme.

The remainder of the paper is organized as follows: Section \ref{Sec:RelatedWork} describes the related work. Section \ref {Sec:SystemModel} presents the system model. Our problem is formulated in Section \ref{Sec:ProblemFormulation}. The analysis about the JTCSP is given in Section \ref{Sec:GameAnalysis}, and a solution is proposed in Section \ref{Sec:Solution}. Section \ref{Sec:PerformanceEvaluation} presents the simulation results. Section \ref{Sec:Conclusion} concludes this paper.

\section{Related Works} \label{Sec:RelatedWork}

Transportation is the most important task of PV systems \cite{zhu2015PublicVehicle} \cite{zhu2013lane-change}, wherein path planning strategies \cite{zhu2015RidesharingComfort} directly determine the quality of service (QoS) of passengers. Some path planning solutions in PV systems or PV-like systems (e.g., taxi sharing systems) have been proposed. Zhu \etal \cite{zhu2016PublicVehicle} propose a path planning strategy with balanced QoS (short waiting time and limited detour) in PV systems. Atasoy \etal \cite{atasoy2015FMOD} propose a flexible mobility-on-demand system utilizing three services: taxi, shared-taxi and mini-bus, which can balance consumer surplus and the operator's profits. Herbawi \etal \cite{herbawi2012ridematching} propose a genetic algorithm to solve the ride-matching problem with time windows in ride-sharing. Fagnant \etal \cite{fagnant2014shared-autonomous-vehicle} present a customer waiting time model of shared autonomous vehicles, and find that it can save 10 times the number of cars needed for self-owned vehicle trips. Jung \etal \cite{jung2016Shared-Taxi-Dispatch} propose a dynamic dispatching algorithm in taxi sharing systems using the hybrid simulated annealing aiming at minimizing the total travel time and maximizing the profits. However, the above works do not consider the energy demands of vehicles.

To ensure the long-term transportation of EVs, charging strategies with low costs are required. Recently, the charging problem of EVs becomes an important topic, e.g., the charging cost minimization problem \cite{rivera2016EV-aggregator} \cite{franco2014milp-ev-charging} and the waiting time minimization problem \cite{zhu2014charging-scheduling} \cite{zhang2014EV-charging-scheduling}. Rivera \etal \cite{rivera2016EV-aggregator} propose an optimization framework for achieving computational scalability based on the alternating directions method of multipliers with two objectives, valley filling and cost-minimal charging with grid capacity constraints. Franco \etal \cite{franco2014milp-ev-charging} present a mixed-integer linear programming model to minimize the total energy costs, which corresponds to the EV charging coordination problem in electrical distribution systems. Zhu \etal \cite{zhu2014charging-scheduling} propose two centralized algorithms to minimize the total charging time, i.e., EV travelling time, the queuing time, and the actual charging time at charging stations. Zhang \etal \cite{zhang2014EV-charging-scheduling} use independent Markov processes to model the uncertainty of the arrival of EVs, the intermittence of the renewable energy, and the variation of grid power prices, and propose a Markov decision process framework to minimize the mean waiting time for EVs under the long-term constraint on the cost. Weerdt \etal \cite{de2016EV-routing} propose a routing policy that minimizes EVs' expected journey time including the waiting time at charging stations.

Game theoretic approaches are applied to transportation systems to optimize traffic flows, travel time, etc. Groot \etal \cite{groot2015reverse-stackelberg} propose three approaches based on reverse Stackelberg game to reach a system-optimal distribution of traffics on freeway routes, e.g., to minimize the total trip time and to reduce traffic emissions in urban areas. The optimization of the travel time and the traffic flows at intersections is a key issue since it is one of the major bottlenecks for urban traffic congestions. To reduce the average travel time delay at uncontrolled intersections, Elhenawy \etal \cite{elhenawy2015game-traffic} propose a chicken-game-based algorithm for controlling autonomous vehicle movements equipped with cooperative adaptive cruise control systems. Bui \etal \cite{bui2017game-traffic-light} propose a Cournot model and a Stackelberg model to optimize the traffic flows for smart traffic light control at intersections. Farokhi \etal \cite{farokhi2015truck-game} adopt an atomic congestion game model to capture the interactions between the car traffic and the truck platooning incentives.

Several novel pricing mechanisms and scheduling strategies have been proposed in on-demand mobility systems. Drwal \etal \cite{drwal2017PricingMechanism} propose pricing mechanisms to balance the demands of different parking stations, reduce the cost of manual relocations of vehicles, and maximize the operator's revenue. However, this solution only considers the transportation performance, and does not consider the energy demands of vehicles from smart grids. Rigas \etal \cite{rigas2015EV-scheduling} study the scenario where EVs are hired in on-demand mobility systems and there exist multiple battery swap facility points. Its aim is to maximize the number of passengers that are served. They focus on the transportation performance with limited travel ranges and simplify the charging problem as swapping batteries in facility points, ignoring the effects of the real-time pricing of smart grids, e.g., how to reduce the charging costs with the guarantee of providing transportation services. Our solution proposed in this paper considers the joint transportation and charging strategies of vehicles over one day and multiple days under the scenario of real-time electricity pricing.

\section{System Model}\label{Sec:SystemModel}

\begin{table}[tbp]
  \centering
  \caption{Variables and notations}
  \label{Tab:VariablesNotations}
  \begin{tabular}{ l l }
  \hline
  $J$ & number of all PVs in the city. \\
  $c$ & battery capacity of PVs. \\
  $r$ & charged energy in a time slot (an hour be default).\\
  $\mathcal{I}$ & set of all PVGs in time slot $t$.\\
  $I$ & number of all PVGs in time slot $t$ in the city. \\
  $m_{i, t}$ & number of unfully charged PVs in region $i$ in time slot $t$.\\
  $\mathcal{T}$ & time slot set. \\
  $t$ & time slot in $\mathcal{T}$. \\
  $x_{i, t}$& strategy of PVG $i$ in time slot $t$. \\
  $d_{i, t}$ & transportation demands of PVG $i$ in time slot $t$. \\
  $d_t$ & total transportation demands of all PVGs in time slot $t$. \\
  $E^{\text{r}}_t$ & total remaining energy of all PVs in time slot $t$. \\
  $p_t$& real-time electricity price in time slot $t$.\\
  $u_{i, t}$ & utility function of PVG $i$ in time slot $t$. \\
  $\Omega_{i, t}$ & feasible strategy set of PVG $i$ in time slot $t$. \\
  $\Omega_t$  & feasible strategy sets of all PVGs in time slot $t$.\\
  $\mathbf{x}_t$   & strategies of all PVGs in time slot $t$. \\
  $\mathbf{u}_t$ & objective function vector of all PVGs in time slot $t$. \\
  $E^{\text{-}}_t$ & consumed energy of all PVs in time slot $t$. \\
  $E^{\text{+}}_t$ & charging demands of all PVs in time slot $t$. \\
  $n_{i, t}$ & number of transportation PVs in region $i$ in time slot $t$.\\
  $a_{i, t}$ & number of all PVs in region $i$ in time slot $t$.\\
  $f_{i, t}$ & number of fully (or near fully) charged PVs in region $i$ \\
  & in time slot $t$. \\
  $\phi_{i, t}$ & number of transportation PVs of PVG $i$ in time slot $t$.\\
  $\psi_{i, t}$ & number of charging PVs of PVG $i$ in time slot $t$.\\
  \hline
  \end{tabular}
\end{table}

We first describe the PV model, and then present the smart grid model. All the important variables and notations in this article are summarized in Table~\ref{Tab:VariablesNotations}.

\subsection{PV Model}

We assume all $J$ PVs in a city have identical battery capacity $c$ (kwh) and identical charged energy $r$ (kwh) in a time slot (an hour by default). To explore the time-dependent transportation and charging patterns, we introduce time slot set $\mathcal{T} = \{ 0, 1, \ldots, T - 1\}$. The PVs are considered fully charged if their remaining energy is greater than $(c - r)$ since the charging time is less than a time slot, which is too short, and the others are unfully charged. Fully charged PVs can only choose to serve passengers. While unfully charged PVs have two choices, transportation or charging. To ensure the transportation of the next time slot, at each time the remaining energy of a PV should not be less than the consumed energy in a time slot.

The city is divided into $I$ regions, and in each region $i$, $m_{i, t}$ denotes the number of unfully charged PVs in time slot $t$ and they form public vehicle group (PVG) $i$. All PVGs in time slot $t$ are denoted by a set $\mathcal{I}$ with the number $I$. We see that $\mathcal{I}$ and $I$ may vary over time and should be updated in each time slot, while $J$ does not change.

\textbf{PVG Transportation Model}: The strategy of PVG $i \in \mathcal{I}$ in time slot $t$ is denoted by $x_{i, t} \in [0, 1]$, which is the ratio of PVs in PVG $i$ that will be used for transportation, and the remaining portion will charge their batteries. In each time slot $t$, the cloud calculates the transportation demands of PVG $i$ denoted by $d_{i, t}$, i.e., the number of transportation PVs. The total transportation demands of all PVGs in time slot $t$ is $d_t = \sum_i d_{i, t}$.

\textbf{PVG Charging Model}: We assume that each charging PV will charge batteries with the identical charged energy $r$ in a time slot. For PVG $i$, the charged energy in the time slot $[t, t + 1]$ are $r \, m_{i, t} \, (1 - x_{i, t})$.

Discharging in peak electricity price periods can bring more benefits for EVs. However, we do not consider discharging in PV systems for several reasons. 1) The most important actions of PVs are transportation and charging since transportation brings profits and charging ensures the ability of making profits. 2) In general, the transportation peak time coincides with the electricity price peak time according to the taxi trip data \cite{NewYorkCityTaxiData} and grid data \cite{SmartGridDataNewYork} in New York City. The urgent action at the energy consumption peak time \cite{Yang2014EV-charging} is charging not discharging. Thus, if we consider discharging, the transportation demands may not be fully satisfied. 3) Even we consider discharging, the profits through V2G are 90$\sim$4,000 US dollars per year per vehicle \cite{yilmaz2013impact-V2G}, which can not match with that from transportation services. Due to the above reasons, we assume that PVs only have two actions: transportation and charging.

\subsection{Smart Grid Model}

Generally, smart grids serve primary consumers such as industries, houses, and offices with high priority. After meeting the demands of the primary consumers, smart grids wish to sell energy to secondary users, e.g., PVs. We assume that, the cloud of PV systems can requests a certain amount of energy from smart grids. In the electricity market, the real-time prices (RTPs) \cite{giraldo2016RTP} vary with the total demands, which can reduce the peak-to-average load ratio through encouraging consumers to shift their usages to off-peak hours. Let the RTP in time slot $t$ be $p_t$. The electricity pricing model is
\begin{eqnarray}
  p_t &=& \alpha_0 \left( \frac{L_t}{C_0} \right)^{k_0},  \\
  L_t &=& \sum_i L_{i, t},
\end{eqnarray}
where $\alpha_0$ and $k_0$ are predefined pricing constants by the smart grids, $C_0$ is the capacity of the electricity markets \cite{wang2009optimizing-energy-use}, $L_{i, t}$ is the electricity load of region $i$ in time slot $t$, and $L_t$ is the total electricity load in time slot $t$, including that from industries, houses, offices and secondary users such as EVs. It has been proven that the overall energy costs are minimal when electricity consumptions are balanced in each equal-size time slot \cite{wang2009optimizing-energy-use}. Then the charging cost of PVG $i$ in the time slot $[t, t + 1]$ is $r \, p_t \,m_{i, t} \, (1 - x_{i, t})$.

\section{Problem Formulation} \label{Sec:ProblemFormulation}

We first present a novel utility model for PVGs, and then formulate the JTCSP.

\subsection{PVG Utility Model}

The utility model for each PVG should reflect its transportation and charging willingness considering the trip demands of passengers, electricity prices, and its own energy states. The utility function of PVG $i$ is formulated as
\begin{eqnarray}
&& u_{i, t} = -(m_{i, t} \, x_{i, t} - d_{i, t})^2 + \alpha_1 \, m_{i, t} \, \text{ln} (2 - x_{i, t})  \nonumber \\
&& ~~~~~~~\, - \alpha_2 \, p_t \, m_{i, t} \, (1 - x_{i, t}) , \label{Eqn:u_i}
\end{eqnarray}
where $\alpha_1$ and $\alpha_2$ are constants. $u_{i, t}$ has the following terms:
\begin{itemize}
  \item{} The term $-(m_{i, t} \, x_{i, t} - d_{i, t})^2$ denotes the transportation utility. Clearly, when the number of transportation PVs of PVG $i$ equals $d_{i, t}$, it obtains the maximum utility. Note that if $m_{i, t} x_{i, t}$ is larger than the critical value $d_{i, t}$, the utility decreases for several reasons such as waiting costs and parking fees, and if it is less than $d_{i, t}$, the utility also decreases since it can not satisfy the transportation demands.
  \item{} The term $\alpha_1 \, m_{i, t} \, \text{ln} (2 - x_{i, t})$ denotes the satisfaction level for charged energy with the weight $\alpha_1$. The charged energy $r$ in a time slot is omitted since it is included in the weight $\alpha_1$, which can also be seen in the third term of the utility model. We adopt a logarithmic utility model to denote the satisfaction level since it can quantify user satisfactions with diminishing returns \cite{lee2015energy-game}, which is widely used in designing the utility for energy consumers, e.g., \cite{liu2017energy-Stackelberg} \cite{park2016contribution-energy-game}. We use $(2 - x_{i, t})$ instead of $(1 - x_{i, t})$ to ensure that it is always positive and the logarithmic function is always available.
  \item{} The term $- \alpha_2 \, p_t \, m_{i, t} \, (1 - x_{i, t})$ denotes the charging fees with the weight $-\alpha_2$.
\end{itemize}

\subsection{Problem Statement}

The transportation and charging strategies of PVs are coordinated rather than directly controlled for the following reasons. 1) This can reduce the cost of data transmission through 4G/5G. During the charging periods, some data should be exchanged between PVs and the cloud, e.g., the remaining energy of batteries, the electricity prices, and some information about charging stations. 2) This can reduce the cost of data storage. PVs are driverless/self-driving EVs, which generate a large amount of data even within a short period, e.g., one GB data per second \cite{williams2013self-driving-car-1GB-per-second}. 3) More computing resources can be assigned to transportation to provide better services for passengers. PV systems need more computing resources in solving path planning problems \cite{Zhu2016TrafficBigData} \cite{zhu2017PathPlanning} based on traffic big data. Moreover, the cloud has to predict in real-time the vehicle speed \cite{jiang2016vehicleSpeedPrediction}, transportation demands \cite{chen2016Traffic-Demand-Forecast} in each region of smart cities.

We explore the cake cutting game \cite{myerson2013game-theory} \cite{chen2010cake-cutting} \cite{branzei2013cake-cutting} to coordinate the transportation and charging strategies for all PVGs. The cake cutting game is one of the most fundamental games for fair division with the aiming of dividing the cake (here, it means the transportation and charging resouces) fairly. The cake cutting game can encapsulate the important problem of allocating heterogeneous resources among multiple players with different preferences. Each PVG has its own transportation demands, and energy states. Therefore, we adopt the cake cutting game to analyze the JTCSP in PV systems.

PVGs are noncooperative since they do not communicate with each other, but they may interact with the cloud and smart grids by the controlled signaling through smart meters. The cloud calculates the charging demands of all PVs $\{ E^{\text{+}}_t \}_{t \in \mathcal{T}}$ ahead of the day, and the transportation demands of each PVG $\{d_{i, t}\}_{i \in \mathcal{I}}$ and the total transportation demands $\{ d_t \}_{i \in \mathcal{I}}$ in each time slot. Then PVGs select their best response strategies $\{x_{i, t}\}_{i \in \mathcal{I}}$ to maximize their utilities.

The optimization problem of the PVG $i$ in time slot $t$ is formulated as
\begin{eqnarray}
&& \max \limits_{x_{i, t}} ~ u_{i, t}, \label{Problem:u_i}\\
&& \text{s.t.}~ r \sum_{i \in \mathcal{I}} m_{i, t} \, (1 - x_{i, t}) = E^{\text{+}}_t, \\
&& ~~~~~ \sum_{i \in \mathcal{I}} m_{i, t} \, x_{i, t} \geq d_t,  \\
&& ~~~~~ x_{i, t} \in [0, 1], \forall i \in \mathcal{I}.
\end{eqnarray}
We see that, for any PVG, the first and second constraints are shared by all PVGs. The first constraint indicates that the total charged energy should be equal to the charging demands $E^{\text{+}}_t$ in time slot $t$ from smart grids. The second constraint indicates that the total number of transportation PVs should not be less than the transportation demands calculated by the cloud, $d_t$. The third constraint indicates the bound of the PVG's strategy $x_{i, t}$. Herein, the second constraint can always hold since the transportation demand $d_t$ is determined by the cloud using a scheduling strategy at the beginning of each time slot, which is described in \textbf{Algorithm 2} of Section~\ref{Sec:Solution}-C. If the transportation demands are too high to be covered for a given set of PVs, the cloud will use more PVs to serve passengers except the charging PVs, and accordingly, the passengers may have to wait for longer time.

Let $\Omega_{i, t}$ be the feasible strategy set of PVG $i$ in time slot $t$ which satisfies the three constraints in Problem (\ref{Problem:u_i}). Let $\Omega_t$ be the feasible strategy sets of all PVGs in time slot $t$, i.e., $\Omega_t := \Omega_{1, t} \times \ldots \Omega_{I, t}$.

We formulate the JTCSP based on PVGs rather than individual PVs for several reasons: 1) This can reduce the operating complexity especially in large cities with a large number PVs, since the number of PVGs is much less than that of PVs. 2) This can save a lot of computing resources such that the energy management costs can be reduced.

The JTCSP is a cake cutting game \cite{myerson2013game-theory} \cite{chen2010cake-cutting}, and in time slot $t$ it is defined by
\begin{equation}\label{Eqn:GameDefinition} \nonumber
( \mathcal{I}, \Omega_t, \mathbf{x}_t, \mathbf{u}_t),
\end{equation}
where
\begin{itemize}
  \item $\mathcal{I}$ denotes the players (PVGs) in time slot $t$ in the cake cutting game;
  \item $\Omega_t := \Omega_{1, t} \times \ldots \Omega_{I, t}$ is the strategy set of all PVGs in time slot $t$;
  \item $\mathbf{x}_t := ( x_{i, t} )_{i \in \mathcal{I}}$ denotes the strategies of all PVGs in time slot $t$;
  \item $\mathbf{u}_t := ( u_{i, t} )_{ i \in \mathcal{I}}$ denotes the objective function vector of all PVGs to maximize in time slot $t$.
\end{itemize}

\section{Problem Analysis} \label{Sec:GameAnalysis}

The first and second constraints in Problem (\ref{Problem:u_i}) imply that the action of a PVG is constrained by the actions of other PVGs, which are known as shared/coupled constraints \cite{kulkarni2014shared-constraint-game}. The games with shared constraints bring particular complexity to tackle. In this section, we prove the existence and uniqueness of equilibrium for our JTCSP.

\begin{lemm} \label{lemm:At-least-one-equilibrium}
An equilibrium point exists in the JTCSP for PV systems.
\end{lemm}
\begin{proof}
Since all PVGs have the same shared constraints, the JTCSP is a generalized Nash equilibrium problem (GNEP). The GNEP extends the classical Nash equilibrium problem by assuming that each player's feasible strategy set can depend on the rival players' strategies. The utility functions of PVGs are continuous and concave and their strategy sets are closed and convex. Therefore, the JTCSP is a concave $n$-person game. According to Rosen's work (Theorem 1 in \cite{rosen1965existence-uniqueness-equilibrium}), ``a Nash equilibrium point exists for every concave $n$-person game", Lemma \ref{lemm:At-least-one-equilibrium} is obtained.
\end{proof}

From Lemma \ref{lemm:At-least-one-equilibrium}, we know that there may exist multiple Nash equilibria in the JTCSP in PV systems. Now, consider the following optimization problem for each utility function:
\begin{eqnarray}
&& \max \limits_{\mathbf{x}_t} ~ \sigma(\mathbf{x}_t, \mathbf{w}_t) = \max \limits_{\mathbf{x}_t} ~ \sum_{i \in \mathcal{I}} w_{i, t} \, u_{i, t}, \label{Eqn:SumUtility}  \\
&& \text{s.t.}~ \bm{\theta}(\mathbf{x}_t) \leq 0,
\end{eqnarray}
where $w_{i, t}$ is a weight factor, and $\bm{\theta}(\mathbf{x}_t)$ = $[\theta_1(\mathbf{x}_t),$ $\ldots,$ $\theta_{M}(\mathbf{x}_t)]^T$ collects $M$ constraint sets which constitute a set $\mathcal{M} = \{ 1, \ldots, M \}$. Here, $M = 3$, which can be seen from the constraints form Problem (\ref{Problem:u_i}). Denote the Lagrange multiplier vector for PVG $i$ as $\bm{\lambda}_{i, t}$, and $\bm{\lambda}_t = (\bm{\lambda}_{1, t}, \bm{\lambda}_{2, t}, \ldots, \bm{\lambda}_{I, t})^T$.

The generalized Nash equilibrium $\mathbf{x}_t \in \Omega_t$ is called a normalized Nash equilibrium (NNE) with weights if and only if it satisfies the following Karush-Kuhn-Tucker (KKT) \cite{wu2007kkt} conditions:
\begin{eqnarray}
&&\hspace{-0.3in} -w_{i, t}\nabla_{x_{i, t}} u_{i, t} + \bm{\lambda}^T_{i, t} \nabla_{x_{i, t}} \bm{\theta}(x_{i, t}, x_{-i, t}) = 0, \label{Eqn:KKT1} \\
&&\hspace{-0.3in} \bm{\lambda}^T_{i, t} \, \bm{\theta}(x_{i, t}, x_{-i, t}) =  0,  \label{Eqn:KKT2}\\
&&\hspace{-0.3in} \bm{\lambda}_{i, t} \geq 0, \label{Eqn:KKT3} \\
&&\hspace{-0.3in} \bm{\theta}(x_{i, t}, x_{-i, t}) \leq  0, \label{Eqn:KKT4}
\end{eqnarray}
where
\begin{equation}
\bm{\lambda}_{i, t} = \frac{\bar{\bm{\lambda}}_t}{w_{i, t}}, \forall i \in \mathcal{I}. \label{Eqn:lambda}
\end{equation}

\begin{theo} \label{theo:unique-normalized-equilibrium}
A unique normalized Nash equilibrium (NNE) exists for the JTCSP in PV systems.
\end{theo}
\begin{proof}
In PV systems, PVGs aim to maximize their utilities through buying low-cost energy from smart grids and providing transportation services for more profits. The objective function of each player in a jointly convex GNEP is continuously differentiable. According to Rosen's work (Theorem 2 in \cite{rosen1965existence-uniqueness-equilibrium}), ``there exists a unique NNE in concave $n$-player games if the joint utility function $\sigma(\mathbf{x}_t, \mathbf{w}_t) = \sum_{i \in \mathcal{I}} w_{i, t} \, u_{i, t}$ with $\mathbf{w}_t = [w_{1, t}, \ldots, w_{I, t}]$ is diagonally strictly concave". Next, we prove $\sigma(\mathbf{x}_t, \mathbf{w}_t)$ is diagonally strictly concave.

We define $\mathbf{g}(\mathbf{x}_t, \mathbf{w}_t)$ as the pseudogradient for $\sigma(\mathbf{x}_t, \mathbf{w}_t)$: $\mathbf{g}(\mathbf{x}_t, \mathbf{w}_t) =$ $[w_{1, t} \nabla_{x_{1, t}} \, u_{1, t} \, (x_{1, t}),$ $w_{2, t} \nabla_{x_{2, t}} \, u_{2, t} \, (x_{2, t}),$ $\ldots,$ $w_{I, t} \nabla_{x_{I, t}} \, u_{I, t} \, (x_{I, t})]^T$. According to Rosen's work (Theorem 6 in \cite{rosen1965existence-uniqueness-equilibrium}), ``a sufficient condition that $\sigma(\mathbf{x}_t, \mathbf{w}_t)$ be diagonally strictly concave for $\mathbf{x}_t \in \Omega_t$ and $\mathbf{w}_t > 0$ is that the symmetric matrix $[\mathbf{G}(\mathbf{x}_t, \mathbf{w}_t) + \mathbf{G}^T(\mathbf{x}_t, \mathbf{w}_t)]$ be negative definite, where $\mathbf{G}(\mathbf{x}_t, \mathbf{w}_t)$ is the Jacobian with respect to $\mathbf{x}_t$ of $\mathbf{g}(\mathbf{x}_t, \mathbf{w}_t)$". The second derivative on the utility function $u_{i, t}$ in (\ref{Eqn:u_i}) with respect to $x_{i, t}$ is
\begin{equation}
\kappa_{i, t} = \frac{\partial^2 u_{i, t}}{\partial x_{i, t}^2} = -2 \, m^2_{i, t} - \frac{\alpha_1 \, m_{i, t}}{(x_{i, t} - 2)^2} .
\end{equation}
Clearly, $\kappa_{i, t} < 0, \forall i \in \mathcal{I}, t \in \mathcal{T}$. So $u_{i, t}$ is strictly concave. The Jacobian of $\mathbf{g}(\mathbf{x}_t, \mathbf{w}_t)$ with respect to $\mathbf{x}_t$ is
\begin{eqnarray}
\hspace{0.13in} \mathbf{G}(\mathbf{x}_t, \mathbf{w}_t) =
\begin{bmatrix}
w_{1, t} \, \kappa_{1, t} & 0 & \cdots & 0 \\
0 & w_{2, t} \, \kappa_{2, t} & \cdots & 0 \\
\vdots & \vdots &  \ddots & \vdots\\
0 & 0 & \cdots & w_{I, t} \, \kappa_{I, t}
\end{bmatrix}.
\end{eqnarray}
Clearly, $\mathbf{G}(\mathbf{x}_t, \mathbf{w}_t)$ is negative definite, and the matrix $[\mathbf{G}(\mathbf{x}_t, \mathbf{w}_t) + \big( \mathbf{G}(\mathbf{x}_t, \mathbf{w}_t) \big)^T] = 2 \, \mathbf{G}(\mathbf{x}_t, \mathbf{w}_t)$ is also negative definite. So $\sigma(\mathbf{x}_t, \mathbf{w}_t)$ is diagonally strictly concave.

Hence, Theorem \ref{theo:unique-normalized-equilibrium} is proved.
\end{proof}

We see that different $\mathbf{w}_t$s will yield different NNEs. However, the NNE is unique for each fixed $\mathbf{w}_t$. Now and henceforth, we consider the NNE with the identical weights, i.e., $w_{1, t} = w_{2, t} = \ldots = w_{I, t} = 1$, and we get
\begin{equation}
\bm{\lambda}_{i, t} = \bar{\bm{\lambda}}_t, \forall i \in \mathcal{I}.
\end{equation}
We know that, all PVGs have the same Lagrange multipliers for all constraints.

Next, we explore the relationships between a variational inequality problem and the JTCSP. The variational inequality problem VI$(F, \Omega)$ is to find a point $x^* \in \Omega$ such that
\begin{equation}
\langle F(x^*), x - x^* \rangle \geq 0, \forall x \in \Omega,
\end{equation}
where $\Omega$ is a closed and convex set, and $\langle \cdot \rangle$ denotes the inner product, and $F$ is a continuous function.

\begin{theo} \label{theo:NNE-VI}
The variational inequality problem VI$(\mathbf{F}_t, \Omega_t)$ with the constraint $\bm{\theta}(\mathbf{x}_t) \leq 0$ where $\mathbf{F}_t := -(\nabla_{x_{i, t}} u_{i, t})_{i \in \mathcal{I}}$ is monotonic. The solution to VI$(\mathbf{F}_t, \Omega_t)$ is the NNE with identical weights, and is also called a variational equilibrium. Therefore, the solution to VI$(\mathbf{F}_t, \Omega_t)$ is the solution to the JTCSP in PV systems.
\end{theo}

\begin{proof}
The KKT conditions for VI$(\mathbf{F}_t, \Omega_t)$ are
\begin{eqnarray}
&& \mathbf{F}_t + \bar{\bm{\lambda}}^T_t \nabla_{\mathbf{x}_t} \bm{\theta}(\mathbf{x}_t) =  0,  \\
&& \bar{\bm{\lambda}}^T_t \,  \bm{\theta}(\mathbf{x}_t) =  0,   \\
&& \bar{\bm{\lambda}}_t \geq 0,  \\
&& \bm{\theta}(\mathbf{x}_t) \leq  0.
\end{eqnarray}
We can see that, the KKT conditions for VI$(\mathbf{F}_t, \Omega_t)$ are exactly (\ref{Eqn:KKT1})$\sim$(\ref{Eqn:KKT4}) with the same Lagrange multipliers for all PVGs in the JTCSP. If $w_{1, t} = w_{2, t} = \ldots = w_{I, t} = 1$, $-\mathbf{F}_t$ is $\mathbf{g}(\mathbf{x}_t, \mathbf{w}_t)$, and correspondingly, the Jacobian of $-\mathbf{F}_t$ is $\mathbf{G}(\mathbf{x}_t, \mathbf{w}_t)$. From Theorem \ref{theo:unique-normalized-equilibrium}, we know $\mathbf{G}(\mathbf{x}_t, \mathbf{w}_t)$ is negative definite, so $\mathbf{F}_t$ is positive definite. Therefore, $\mathbf{F}_t$ is strictly monotonic. So Theorem \ref{theo:NNE-VI} is proved. The method to  the variational inequality problem can be used to solve the GNEP \cite{facchinei2007GNEP}.
\end{proof}

\section{Proposed Solution} \label{Sec:Solution}

In this section, we first describe the basic idea of our proposed solution for solving the JTCSP in PV systems, and then present two key routines: charging scheduling and transportation scheduling, and finally detail the scheme to obtain the unique NNE.

Why is there an optimal solution to the JTCSP in PV systems? Essentially, both the transportation and charging strategies focus on the profits earned by providing services for passengers: transportation strategies focus on the current profits from passengers, while charging strategies focus on the future profits since the charged energy will be used for transportation to make profits. If we only focus on the current transportation, the future profits may not be fully obtained since some PVs may not have sufficient energy. If we only focus on the current charging, the charging costs can be minimized, while the profits from the current transportation may not be maximized.

\subsection{Basic Idea}

We assume that the trip requests of passengers and RTP of the next day are known in advance. The basic idea of our solution is as follows. 1) The cloud performs the charging scheduling one day ahead to calculate charging demands. 2) In each time slot, the cloud performs the transportation scheduling according to the real-time trip requests to calculate transportation demands. 3) We use a projection method to solve the variational inequality problem and obtain the best response strategies of PVGs. 4) Each PVG selects the a part of PVs to provide transportation services and others to charge batteries according to its best strategy. With respect to the transportation PVs, we use PCI \cite{zhu2016PublicVehicle} to schedule them to serve passengers, and with respect to the charging PVs, they travel to the nearest charging stations to charge batteries.

\subsection{Charging Scheduling}

We assume that the cloud knows the trip requests of one day ahead, therefore, the charging scheduling of PVs can be calculated. The objective of charging scheduling is to minimize the total charging costs considering of energy consumptions of PVs and RTPs. Suppose that the PVs have infinite energy and the energy consumption is positively correlated with the travel distance. We schedule all PVs using the PCI algorithm \cite{zhu2016PublicVehicle} to serve passengers, and then record the travel distance and consumed energy of all PVs in each time slot $\{ E^{\text{-}}_t \}_{t \in \mathcal{T}}$. The process of PCI is as follows: First, all requests will be sorted by their waiting time from the maximum to the minimum, i.e., the passengers with longest waiting time will have the highest scheduling rank. Second, insert each request to the path of each PV and then calculate the insertion cost (total travel distance of PVs) and detour ratio. Third, if we can find the ride-match and path with the minimum insertion cost from all paths within limited detour ratio, schedule the PV to serve the corresponding request, otherwise the request will be put to a waiting list.

\vspace{0.05in}
\vspace{0.05in}
\vspace{0.05in}
\begin{tabular}[tbp]{lp{0.45\textwidth}}
  \toprule
  \textbf{Algorithm 1}: Charging Scheduling One Day Ahead\\
  \toprule
  ~1:~Initialize the remaining energy for all PVs in \\
  ~~~~the city as infinity;\\
  ~2:~Use PCI \cite{zhu2016PublicVehicle} to schedule PVs to serve passengers;\\
  ~3:~Record $\{ E^{\text{-}}_t \}_{t \in \mathcal{T}}$; \\
  ~4:~Calculate $\{ E^{\text{+}}_t \}_{t \in \mathcal{T}}$ by solving Problem~ (\ref{Problem:PredictChargedEnergy}) using \\
  ~~~~a convex solver;\\
  ~5:~\textbf{return} $\{ E^{\text{+}}_t \}_{t \in \mathcal{T}}$;\\
  \bottomrule
\end{tabular}
\vspace{0.05in}
\vspace{0.05in}
\vspace{0.05in}

The optimization problem of charging scheduling is formulated as
\begin{eqnarray}
&&\hspace{-0.2in} \min ~ \sum_{t \in \mathcal{T}} p_t \, E^{\text{+}}_t    \label{Problem:PredictChargedEnergy}, \\
&&\hspace{-0.40in} \text{s.t.}~E^{\text{r}}_{t + 1} =  E^{\text{r}}_t - E^{\text{-}}_t + E^{\text{+}}_t, \forall t \in \{0, \ldots, T - 2\},  \\
&&\hspace{-0.40in}~~~~E^{\text{r}}_t \geq (1 + \rho) \max (E^{\text{-}}_t, J \, e^{\text{min}}), \forall t \in \{1, \ldots, T - 1\},  \\
&&\hspace{-0.40in}~~~~E^{\text{r}}_{T - 1} - E^{\text{-}}_{T - 1} +  E^{\text{+}}_{T - 1} \geq E^{\text{r}}_0,  \\
&&\hspace{-0.40in}~~~~0 \leq E^{\text{+}}_t \geq (J - d_t) \, r, \forall t \in \{0, \ldots, T - 1\}, \\
&&\hspace{-0.40in}~~~~0 \leq E^{\text{r}}_t \leq J \, c, \forall t \in \{0, \ldots, T - 1\},
\end{eqnarray}
where $E^{\text{+}}_t$ and $ E^{\text{-}}_t$ are the charging demands and consumed energy of all PVs in time slot $t$ respectively, $E^{\text{r}}_t$ is the total remaining energy of all PVs in time slot $t$, $\rho > 0$ is a constant, $J$ is the number of all PVs in the city, and $e^{\text{min}}$ is the minimum energy of any PV to travel to the nearest the charging station. The first constraint is on the relationship between the consumed energy, the charging demands (charged energy), and the remaining energy in two consecutive time slots. The second constraint means that the remaining energy in time slot $t$ is at least $(1 + \rho)$ of the consumed energy used on transportation, which ensures the travel of the PVs in the next time slot, at the same time, the remaining energy is no less than the $(1 + \rho)$ of the energy consumed on the travel to charging stations. We assume that the final remaining energy is not less than the initial energy, which is reflected in the third constraint. The fourth constraint implies that the charged energy has a upper bound if all PVs except for the transportation ones decide to charge batteries. The fifth constraint points out that the remaining energy should not exceed the total energy capacities of all PVs. Problem (\ref{Problem:PredictChargedEnergy}) is a standard convex problem, and can be solved by any standard convex solver. The procedure for charging scheduling of one day is summarized in \textbf{Algorithm 1}.

\subsection{Transportation Scheduling} \label{Sec:PredictTransportationDemand}

In time slot $t$, the cloud calculates transportation scheduling, i.e., the number of transportation PVs, $n_{i, t}$, in each region $i$ given the trip requests using a vehicle scheduling scheme. We get the number of unfully charged PVs in region $i$, i.e., the number of PVs in PVG $i$:
\begin{equation}\label{Eqn:PVG-NumPVs}
m_{i, t} = a_{i, t} - f_{i, t},
\end{equation}
where $a_{i, t}$ is the number of all PVs in region $i$ in time slot $t$, and $f_{i, t}$ is the number of fully (or near fully) charged PVs in region $i$ in time slot $t$. In time slot $t$ we use PCI \cite{zhu2016PublicVehicle} to schedule PVs to serve passengers, and the transportation demands of PVG $i$ in time slot $t$ is
\begin{equation}\label{Eqn:NumExpPVs}
d_{i, t} = \max( n_{i, t} - f_{i, t}, 0).
\end{equation}
The total transportation demands in time slot $t$ is
\begin{equation}\label{Eqn:MinNumTransportationPVs}
d_t = \sum_i d_{i, t}.
\end{equation}

\vspace{0.05in}
\vspace{0.05in}
\begin{tabular}[tbp]{lp{0.45\textwidth}}
  \toprule
  \textbf{Algorithm 2}: Transportation Scheduling in Time Slot $t$ \\
  \toprule
  ~1:~Record the initial states of all PVs and passen-\\
  ~~~~gers;\\
  ~2:~Use PCI \cite{zhu2016PublicVehicle} to serve the passengers in the time  \\
  ~~~~slot $[t, t+1]$; \\
  ~3:~Record the number of transportation PVs in each\\
  ~~~~region $\{ n_{i, t} \}_{i \in \mathcal{I}}$;\\
  ~4:~Calculate $d_{i, t}$ and $d_t$ using (\ref{Eqn:NumExpPVs}) and (\ref{Eqn:MinNumTransportationPVs}) respec-\\
  ~~~~tively; \\
  ~5:~Put all PVs and passengers to initial states;\\
  ~6:~\textbf{Return} $d_{i, t}$ and $d_t$; \\
  \bottomrule
\end{tabular}
\vspace{0.05in}
\vspace{0.05in}

\textbf{Algorithm 2} shows the procedure of transportation scheduling in time slot $t$. The positions, remaining energy, energy consuming rates, service states of all PVs, and the service states (pickup or dropoff) of passengers are called initial states, which are recorded in line 1. Line 2 uses the PCI algorithm to serve passengers considering of energy consumptions and energy limits. Line 3 records the number of transportation PVs of all regions. Line 4 calculates the transportation demands $d_{i, t}$ of PVG $i$ and the total transportation demands $d_t$. Line 5 means that all PVs and passengers return to their initial states.

\subsection{Solution to JTCSP}

From Theorem \ref{theo:NNE-VI}, we know that $\mathbf{F}_t$ in VI$(\mathbf{F}_t, \Omega_t)$ is monotonic with respect to its strategy set. So we can use a hyperplane projection method to solve it and the convergence can be guaranteed \cite{facchinei2007GNEP-VI}. Here, we use SSPM (Solodov-Svaiter projection method) \cite{solodov1999projection-method} to solve VI$(\mathbf{F}_t, \Omega_t)$. The projection operator $P_{\Omega_t}$ is defined as
\begin{equation}
P_{\Omega_t}(\mathbf{x}_t) = \arg \min \limits_{\mathbf{x'}_t \in \Omega_t} \| \mathbf{x'}_t - \mathbf{x}_t \|.
\end{equation}
We introduce parameters $\gamma_1 \in (0, 1)$, $\gamma_2 \in (0, 1)$, $\gamma_3 > 1$, $\mu > 0$, and $\eta > 0$. $\mu^{(k)}$ is calculated by
\begin{equation}\label{Eqn:mu}
\mu^{(k)} = \min (\gamma_3 \, \eta^{(k - 1)}, 1).
\end{equation}
The projected residual function is defined as
\begin{equation}\label{Eqn:residual}
\nu(\mathbf{x}^{(k)}_t, \mu^{(k)}) := \mathbf{x}^{(k)}_t - P_{\Omega_t} \big( \mathbf{x}^{(k)}_t - \mu^{(k)} \mathbf{F}_t(\mathbf{x}^{(k)}_t) \big).
\end{equation}
Let $\zeta^{(k)}$ be the smallest nonnegative integer which satisfies
\begin{eqnarray}
&& \langle \mathbf{F}_t(\mathbf{x}^{(k)}_t - \gamma_1^{\zeta^{(k)}} \, \mu^{(k)} \, \nu(\mathbf{x}^{(k)}_t, \mu^{(k)})), \nu(\mathbf{x}^{(k)}_t, \mu^{(k)}) \rangle \nonumber \\
&& ~~~ \geq \frac{\gamma_2}{\mu^{(k)}} \| \nu(\mathbf{x}^{(k)}_t, \mu^{(k)}) \|^2, \label{Eqn:zeta}
\end{eqnarray}
where $\langle \cdot \rangle$ denotes the inner product. $\eta^{(k)}$ and $\mathbf{y}^{(k)}$ are calculated by
\begin{eqnarray}
\eta^{(k)} &=& \gamma_1^{\zeta^{(k)}} \, \mu^{(k)}, \label{Eqn:eta}\\
\mathbf{y}^{(k)} &=& \mathbf{x}^{(k)}_t - \eta^{(k)} \, \nu(\mathbf{x}^{(k)}_t, \mu^{(k)}). \label{Eqn:y}
\end{eqnarray}
The halfspace $\mathbb{H}^{(k)}$ is defined as
\begin{equation}\label{Eqn:halfspace}
\mathbb{H}^{(k)} := \{ \mathbf{x}^{(k)}_t \in \mathfrak{R}^I | \langle \mathbf{F}_t(\mathbf{y}^{(k)}), \mathbf{x}^{(k)}_t - \mathbf{y}^{(k)} \rangle \,\, \leq 0\} .
\end{equation}

\vspace{0.05in}
\vspace{0.05in}
\vspace{0.05in}
\begin{tabular}[tbp]{lp{0.45\textwidth}}
  \toprule
  \textbf{Algorithm 3}: SSPM   \\
  \toprule
  ~1:~Initialize $\mathbf{x}^{(0)}_t \in \Omega_t$, $\eta^{(-1)} > 0$, $\gamma_1 \in (0, 1)$, $\gamma_2 \in$ \\
  ~~~$(0, 1)$, $\gamma_3 > 1$, $\nu(\mathbf{x}^{(0)}_t, \mu^{(0)}) = 1$, $k = 0$. Set a small \\
  ~~~positive value for a bound $\epsilon$; \\
  ~2:~\textbf{repeat} \\
  ~3:~~~Calculate $\mu^{(k)}$ and $\nu(\mathbf{x}^{(k)}_t, \mu^{(k)})$ using (\ref{Eqn:mu}) and \\
  ~~~~~~(\ref{Eqn:residual}) respectively;\\
  ~4:~~~\textbf{if} $\| \nu(\mathbf{x}^{(k)}_t, \mu^{(k)}) \| < \epsilon$ \\
  ~5:~~~~~\textbf{break}; \\
  ~6:~~~\textbf{else}\\
  ~7:~~~~~Calculate the smallest nonnegative integer $\zeta^{(k)}$  \\
  ~~~~~~~~satisfying (\ref{Eqn:zeta}), and then calculate $\eta^{(k)}$ and $\mathbf{y}^{(k)}$\\
  ~~~~~~~~using (\ref{Eqn:eta}) and (\ref{Eqn:y}) respectively;\\
  ~8:~~~~~Calculate the halfsapce $\mathbb{H}^{(k)}$ using (\ref{Eqn:halfspace}), and \\
  ~~~~~~~~then calculate $\mathbf{x}^{(k + 1)}_t$ using (\ref{Eqn:x-next});\\
  ~9:~~~~~$k \leftarrow k + 1$;\\
  10:~~~\textbf{end if} \\
  11:~\textbf{until} $\| \nu(\mathbf{x}^{(k)}_t, \mu^{(k)}) \| < \epsilon$; \\
  12:~\textbf{return} $\mathbf{x}_t$; \\
  \bottomrule
\end{tabular}
\vspace{0.05in}
\vspace{0.05in}
\vspace{0.05in}

\vspace{0.05in}
\begin{tabular}[tbp]{lp{0.45\textwidth}}
  \toprule
  \textbf{Algorithm 4}: Joint Transportation and Charging~~~~~~~~~\\
  ~~~~~~~~~~~~~~~~~Scheduling (JTCS)\\
  \toprule
  ~1:~Perform charging scheduling one day ahead using \\
  ~~~~\textbf{Algorithm~1}; \\
  ~2:~\textbf{for} $t \in \mathcal{T}$ \\
  ~3:~~~Update each PVG $i$, $\forall i \in \mathcal{I}$; \\
  ~4:~~~Perform transportation scheduling in time slot $t$ \\
  ~~~~~~using \textbf{Algorithm~2}; \\
  ~5:~~~The cloud announces the charging demands and \\
  ~~~~~~transportation demands;\\
  ~6:~~~Calculate the best response strategies $\mathbf{x^*}_t = $\\
  ~~~~~~$\{ x_{1, t}^*, \ldots, x_{I, t}^* \}$ of all PVGs using SSPM \\
  ~~~~~~~(\textbf{Algorithm~3}) given the cloud's strategies;\\
  ~7:~~~\textbf{for} PVG $i \in \mathcal{I}$ \\
  ~8:~~~~~$\phi_{i, t} \leftarrow \lceil m_{i, t} x_{i, t}^* \rceil$;\\
  ~9:~~~~~$\psi_{i, t} \leftarrow m_{i, t} - \phi_{i, t}$;\\
  10:~~~\textbf{end} \\
  11:~\textbf{end} \\
  \bottomrule
\end{tabular}
\vspace{0.05in}
\vspace{0.05in}

Then $\mathbf{x}^{(k + 1)}_t$ is obtained by projecting $\mathbf{x}^{(k)}_t$ onto the intersection of its feasible set $\Omega_t$ and the halfspace $\mathbb{H}^{(k)}$:
\begin{equation}\label{Eqn:x-next}
\mathbf{x}^{(k + 1)}_t = P_{\Omega_t \cap \mathbb{H}^{(k)}} (\mathbf{x}^{(k)}_t).
\end{equation}
The procedure of SSPM \cite{solodov1999projection-method} is shown in \textbf{Algorithm 3}. Line 1 is initialization. Lines 2$\sim$11 constitute an iterative process until the projected residual is less than a preset bound. Lines 3 and 8 imply that only two projections are needed in each iteration.

\textbf{Algorithm 4} shows the overall algorithm for the JTCSP. Line 1 uses \textbf{Algorithm~1} to perform the charging scheduling in one day. The ``for" loop between lines 2$\sim$11 indicates that in each time slot the SSPM method will be executed once. Line 3 means that in each time slot all PVGs will be updated, since the unfully charged PVs in each region may change. Line 6 uses SSPM (\textbf{Algorithm 3}) to get the best response strategies of PVGs. In lines 8 and 9, $\phi_{i, t}$ and $\psi_{i, t}$ are the numbers of transportation and charging PVs respectively of PVG $i$ in time slot $t$. Finally, for PVG $i$, the $\phi_{i, t}$ PVs with the maximum remaining energy are selected to provide transportation services and others to charge batteries.

\subsection{Comparison with Existing Approaches}

There exist two cakes in the JTCSP in PV systems, a transportation cake and a charging cake. The proposed JTCS algorithm considers how to cut the two cakes to allocate transportation and charging resources among PVGs. However, if we only consider cutting the transportation or charging cake, the process of cutting the other cake will be in a disorderly state.

In the optimization problem (\ref{Problem:u_i}) of PVG $i$ in PV systems, we see that the first constraint implies the condition in cutting the charging cake, and the second constraint implies the condition in cutting the transportation cake. The utility model in (\ref{Eqn:u_i}) reflects the transportation and charging utilities. However, if we only consider cutting the transportation cake, the utility model of PVG $i$ is
\begin{equation}
v_i = -(m_{i, t} x_{i, t} - d_{i, t})^2,
\end{equation}
and the corresponding optimization problem is
\begin{eqnarray}
&& \max \limits_{x_{i, t}} ~ v_{i, t}, \label{Problem:v_i}\\
&& \text{s.t.}~\sum_{i \in \mathcal{I}} m_{i, t} \, x_{i, t} \geq d_t,  \\
&& ~~~~~ x_{i, t} \in [0, 1], \forall i \in \mathcal{I}.
\end{eqnarray}
We see that, Problem (\ref{Problem:v_i}) have a different objective function and different constraints compared with Problem (\ref{Problem:u_i}). So if we only consider cutting transportation or charging cake, how to cut the other cake will be neglected.

The proposed JTCS algorithm always converge to the NNE at each timeslot in the JTCSP, since the projection method SSPM in the JTCS algorithm always converges to the optimal solution to VI$(\mathbf{F}_t, \Omega_t)$. In \cite{solodov1999projection-method}, it is shown that, SSPM converges to a solution of the variational inequality problem under the only assumption that its function is continuous and monotonic. In Theorem \ref{theo:NNE-VI}, we know that $\mathbf{F}_t$ is continuous and monotonic, therefore, the solution to VI$(\mathbf{F}_t, \Omega_t)$ is the NNE with identical weights.

\section{Performance Evaluation} \label{Sec:PerformanceEvaluation}

In this section, we first describe the simulation settings, and then present the simulation results based on real data sets.

We compare the JTCS algorithm with a heuristic solution: transportation with greedy charging (TGC), where we first use the PCI algorithm proposed in \cite{zhu2016PublicVehicle} to provide transportation services and then schedule all the other unfully charged PVs to charge batteries until fully charged.

\subsection{Simulation Settings}

\textbf{PV Setting}: We use 500 EVs, Yutong E7 \cite{Yutong_E7_Homepage} (China), to study the transportation and charging patterns of PVs, although they are not self-driving vehicles now. The number of seats of Yutong E7 is 10$\sim$30, and here we assume that each PV has 16 seats. Each PV has electricity capacity of 45 kwh, and the maximum travel distance with fully charged battery is 150 km, i.e., it consumes 0.3 kwh each km. Under 220V voltage, the batteries can be fully charged within 8 hours, i.e., the charged energy one hour is 5.625 kwh. We assume that the initial remaining energy of all PVs follows a uniform distribution over [32, 41] kwh. Assuming that all PVs travel along the shortest path between any two positions (origins or destinations of requests) in PCI algorithm with the identical speed 30 km/h.

\begin{figure}
  \centering
  \includegraphics[height=0.25\linewidth,width=0.98\linewidth]{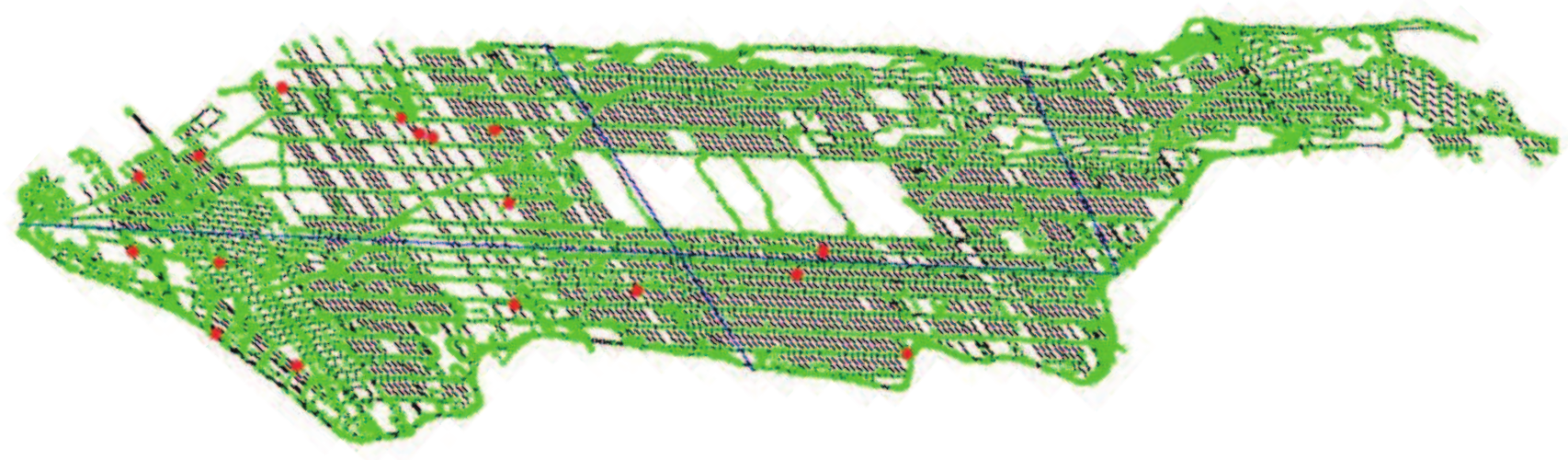}
  \caption{Manhattan in New York City.}
  \label{Fig:Manhattan}
\end{figure}
\begin{figure*}[htbp]
  \begin{minipage}[b]{0.33\linewidth}
    \centering
    \includegraphics[height=0.80\linewidth,width=0.99\linewidth]{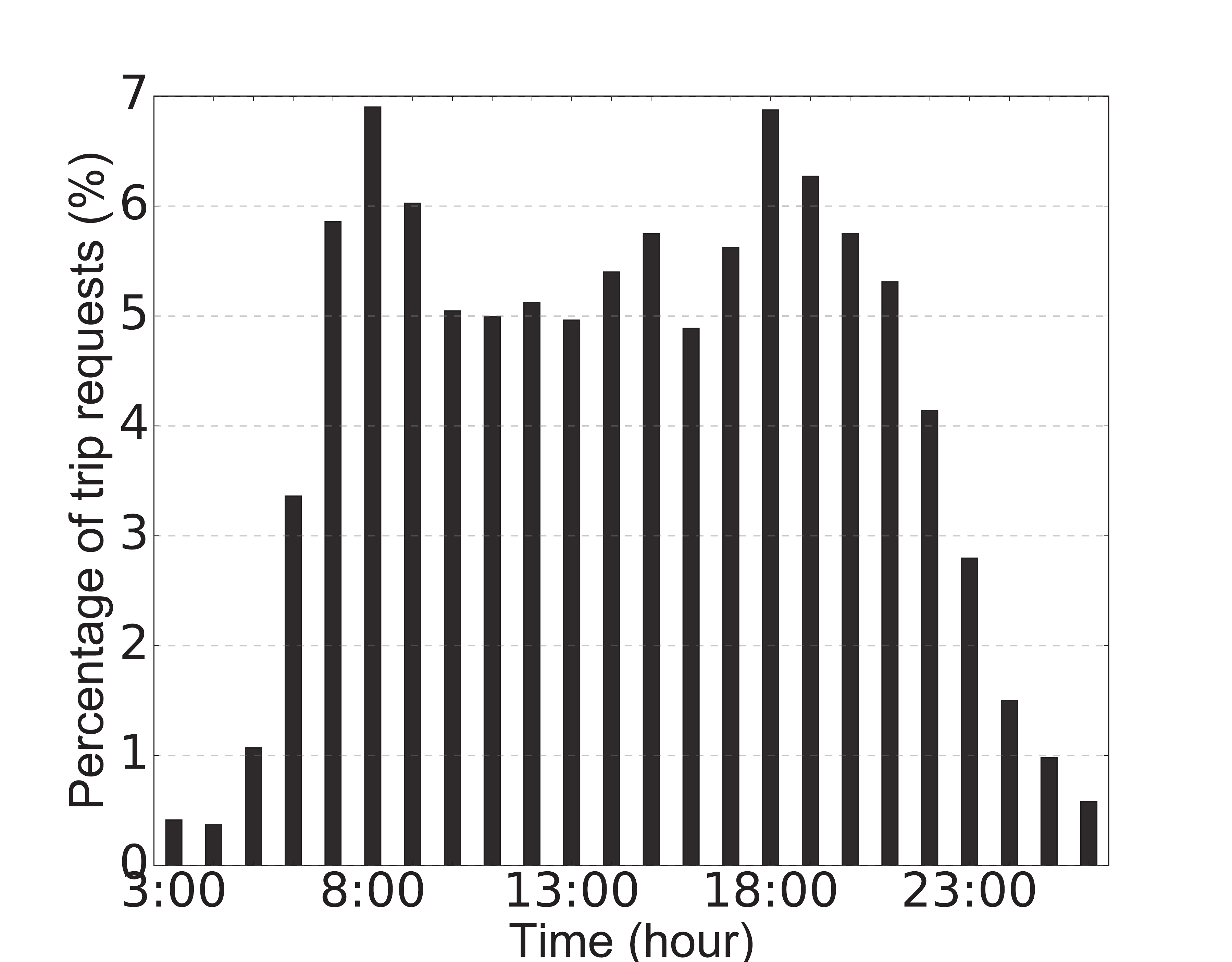}
    \caption{Distribution of trip requests in Manhattan.}
    \label{Fig:DistributionOfRrequests}
  \end{minipage}  \hspace{0.02in}
  \begin{minipage}[b]{0.33\linewidth}
    \centering
    \includegraphics[height=0.80\linewidth,width=0.99\linewidth]{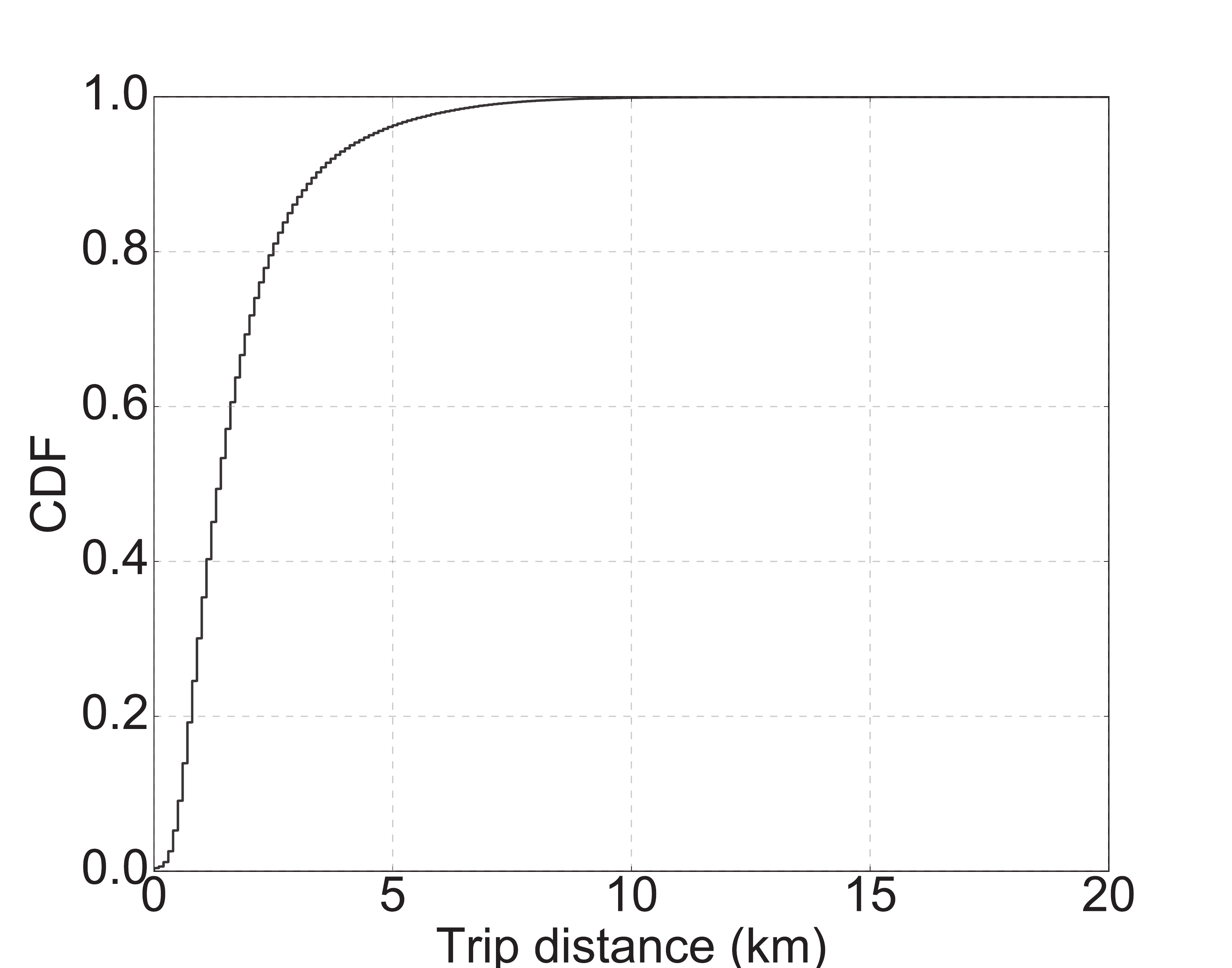}
    \caption{CDF of trip distance in Manhattan.}
    \label{Fig:CDF_MahattanTaxiDataDistTrip}
  \end{minipage}
  \begin{minipage}[b]{0.33\linewidth}
    \centering
    \includegraphics[height=0.80\linewidth,width=0.99\linewidth]{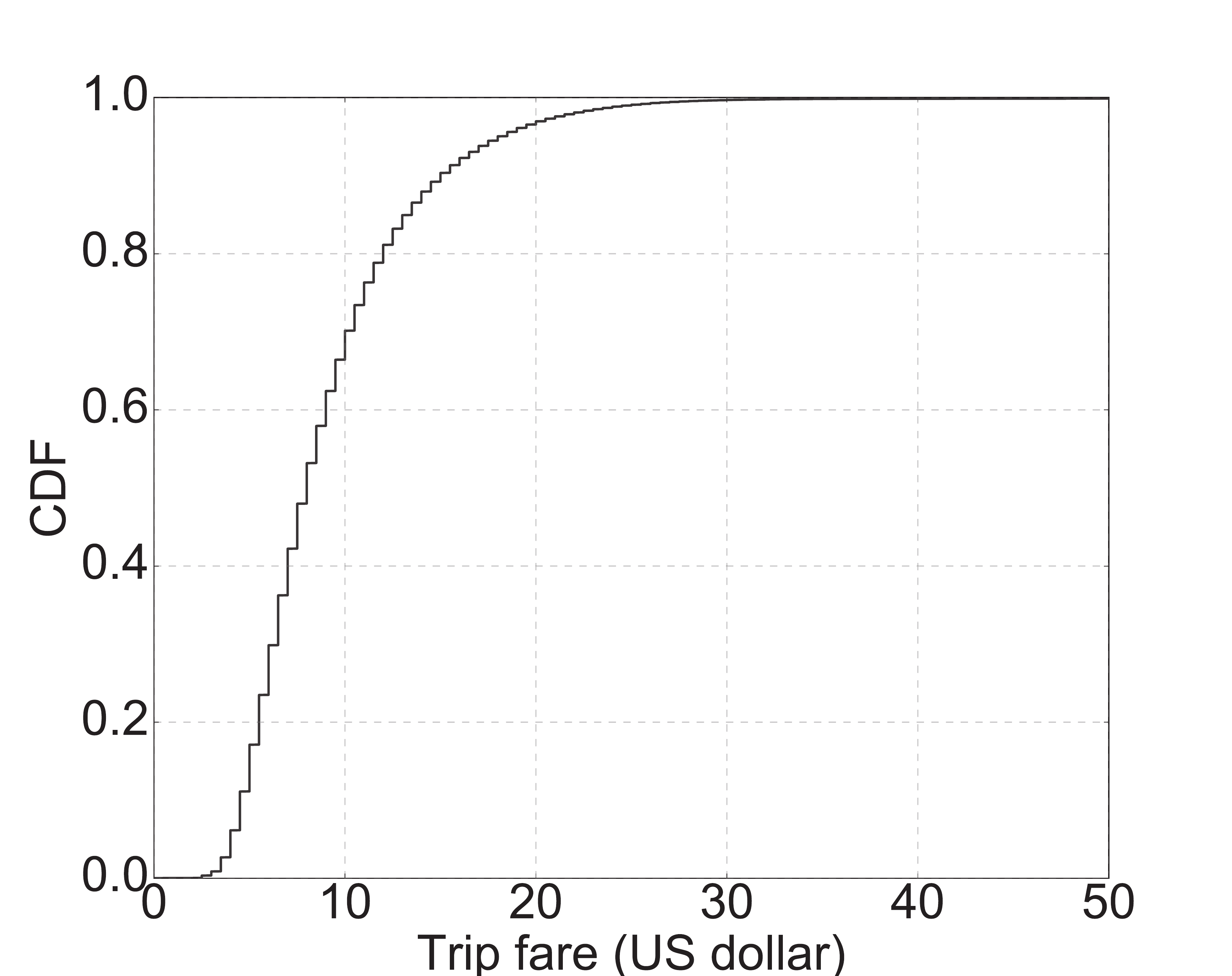}
    \caption{CDF of trip fares in Manhattan.}
    \label{Fig:CDF_MahattanTaxiDataFare}
  \end{minipage}
\end{figure*}

\textbf{Road Map Data}: Fig.~\ref{Fig:Manhattan} shows the map of Manhattan with about 60 km$^2$ in New York City, where black lines are roads, green points are nodes, and red points are charging stations, and the blue lines divide the map into five regions. The road map is extracted through the openstreetmap \cite{Openstreetmap}, and six types of ways are selected: primary, secondary, tertiary, motorway, motorway\_link, and residential, and others such as trunk, unclassified, are ignored. Finally, 3,900 ways and 29,792 nodes are filtered. The longitudes and latitudes of charging stations are extracted from the Google map, and then they are moved to the nearest nodes on roads.

\begin{figure*}[htbp]
  \begin{minipage}[b]{0.33\linewidth}
  \centering
  \includegraphics[height=0.80\linewidth,width=0.99\linewidth]{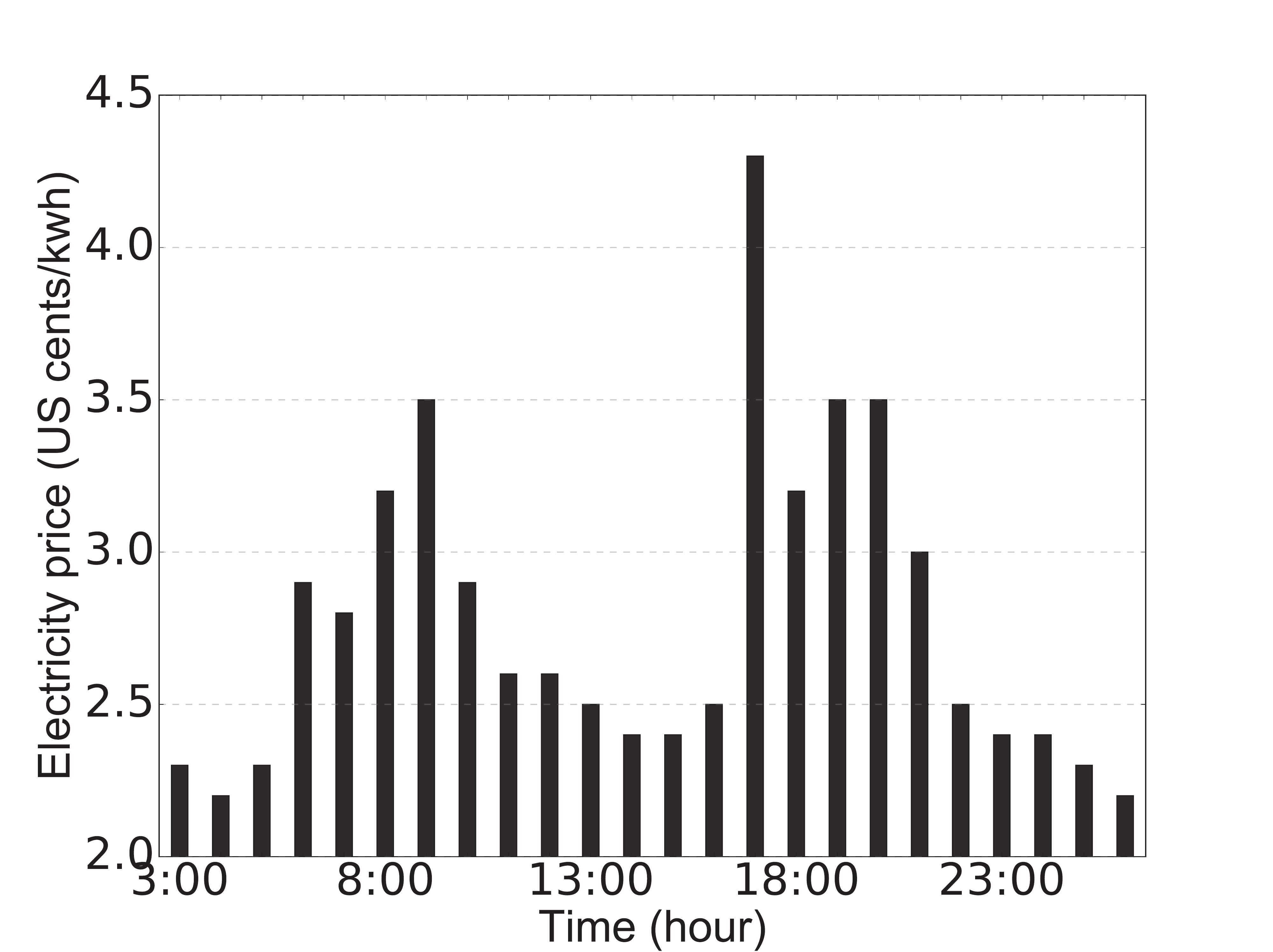}
  \caption{Real-time price (November 27, 2016)}
  \label{Fig:RealTimePrice}
  \end{minipage}  \hspace{0.02in}
  \begin{minipage}[b]{0.33\linewidth}
  \centering
  \includegraphics[height=0.80\linewidth,width=0.99\linewidth]{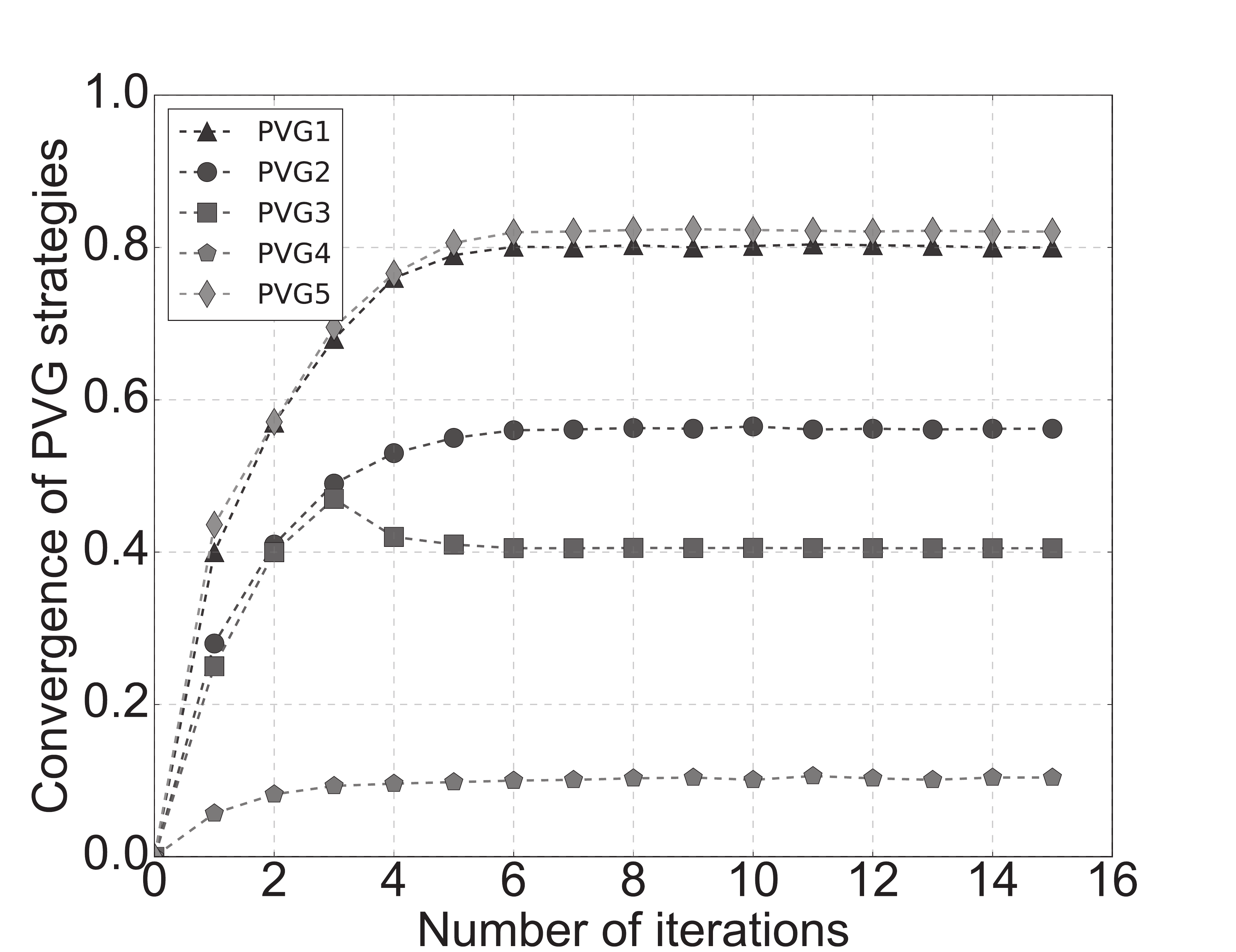}
  \caption{Convergence of PVG strategies.}
  \label{Fig:ConvergencePVG}
  \end{minipage}
  \begin{minipage}[b]{0.33\linewidth}
  \centering
  \includegraphics[height=0.80\linewidth,width=0.99\linewidth]{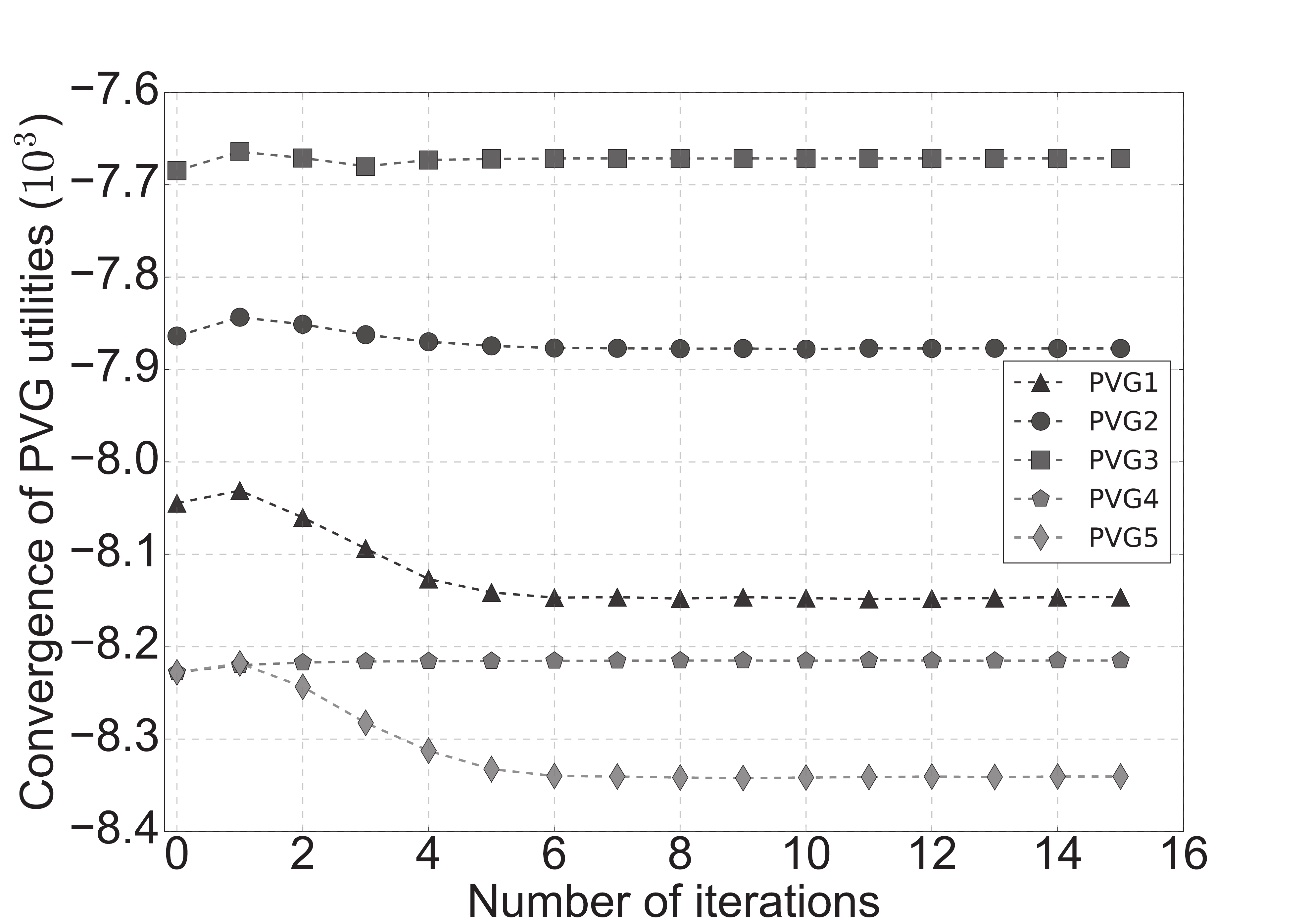}
  \caption{Convergence of PVG utilities.}
  \label{Fig:ConvergenceUtility}
  \end{minipage}
\end{figure*}

\textbf{New York City Taxi Data}: We use the taxi data set (yellow records) of the New York City during January 1$\sim$31, 2016 \cite{NewYorkCityTaxiData}. Each record contains several useful fields for our study, including passenger count, pickup time, dropoff time, trip distance, latitudes/longitudes of origins, latitudes/longitudes of destinations, fares, taxes, tips, and total payment. There are 134,721 requests whose origins and destinations both are in Manhattan on January 5 (Tuesday), 2016. Fig.~\ref{Fig:DistributionOfRrequests} shows the distribution of these trip requests (only in Manhattan) in each hour on January 5, 2016. The time begins at 3:00 for three reasons. 1) This time is one of the most important shift handover time in taxi companies in many cities \cite{tian2016ChargingStationRecommendation}. 2) The number of trip requests of this time is almost the minimum of one day \cite{zhu2016PublicVehicle} \cite{atasoy2015FMOD}. 3) The base loads and energy demands at this time are both almost the minimum in one day \cite{SmartGridDataNewYork} \cite{ma2013EV-charging}. Fig.~\ref{Fig:CDF_MahattanTaxiDataDistTrip} and Fig.~\ref{Fig:CDF_MahattanTaxiDataFare} show the CDF of trip distance and trip fares (tips and taxes are not included) respectively in Mahattan on January 5, 2016. We see that, about 70\% of trip distance is less than 2 km, and 70\% of trip fares are less than 10 US dollars. To make the performance more stable, we only choose the trip requests with travel distance no less than 2 km, and the number of trip requests is reduced to 41,341.

\textbf{Smart Grid Data}: The online data set of RTPs is shown in Fig.~\ref{Fig:RealTimePrice}, which is provided by the Commonwealth Edison Company \cite{RealTimePrice-Commonwealth}. We see that the RTP morning and evening peak time is 9:00 and 17:00 respectively, and the RTP evening peak is much higher than the morning peak.

The parameter settings are summarized in Table~\ref{Tab:ValuesParameters}. $\alpha_1$ and $\alpha_2$ should reflect the weights of satisfaction level of charging and the charging fees, and should be tested to represent the utilities of PVGs. As we have described in the previous section, $\gamma_1 \in (0, 1), \gamma_2 \in (0, 1), \gamma_3 > 1, \mu > 0, \eta > 0$, therefore, we set $\gamma_1 = 0.4$, $\gamma_2 = 0.5$, $\gamma_3 = 1.5$, $\mu = 1$, $\eta = 1$. $e^{\text{min}}$ denotes the minimum energy of any PV to travel to the nearest the charging station, and 3 kwh is generally enough for one PV in Manhattan since it can travel at least 10 km according to our assumptions. We set $\rho = 0.2$, which is generally enough for the energy in the next time slot. $\epsilon = 10^{-3}$ is a widely accepted bound in SSPM.

\begin{table}[tbp]
  \centering
  \caption{Values of parameters in simulation settings}
  \label{Tab:ValuesParameters}
  \begin{tabular}{ c c c c c c    }
  \hline
  $\alpha_1$ & $\alpha_2$ & $\eta$ & $\gamma_1$ & $\gamma_2$ & $\gamma_3$   \\
  20 & 5 & 1 & 0.4 & 0.5 & 1.5     \\
  \hline
  $\mu$ & $e^{\text{min}}$ & $\rho$  & $\epsilon$   &     &  \\
  1 & 3& 0.2 &  $10^{-3}$ &  & \\
  \hline
  \end{tabular}
\end{table}

\subsection{Results}

We present the performance of JTCS and TGC in terms of eight metrics: convergence of PVG strategies, convergence of PVG utilities, number of transportation PVs in each hour, consumed energy in each hour, charged energy in each hour, energy payment in each hour, average energy price and total payment, and remaining energy of all PVs. Finally, the scalability of JTCS and TGC is discussed.

\begin{figure*}[htbp]
  \begin{minipage}[b]{0.33\linewidth}
  \centering
  \includegraphics[height=0.80\linewidth,width=0.99\linewidth]{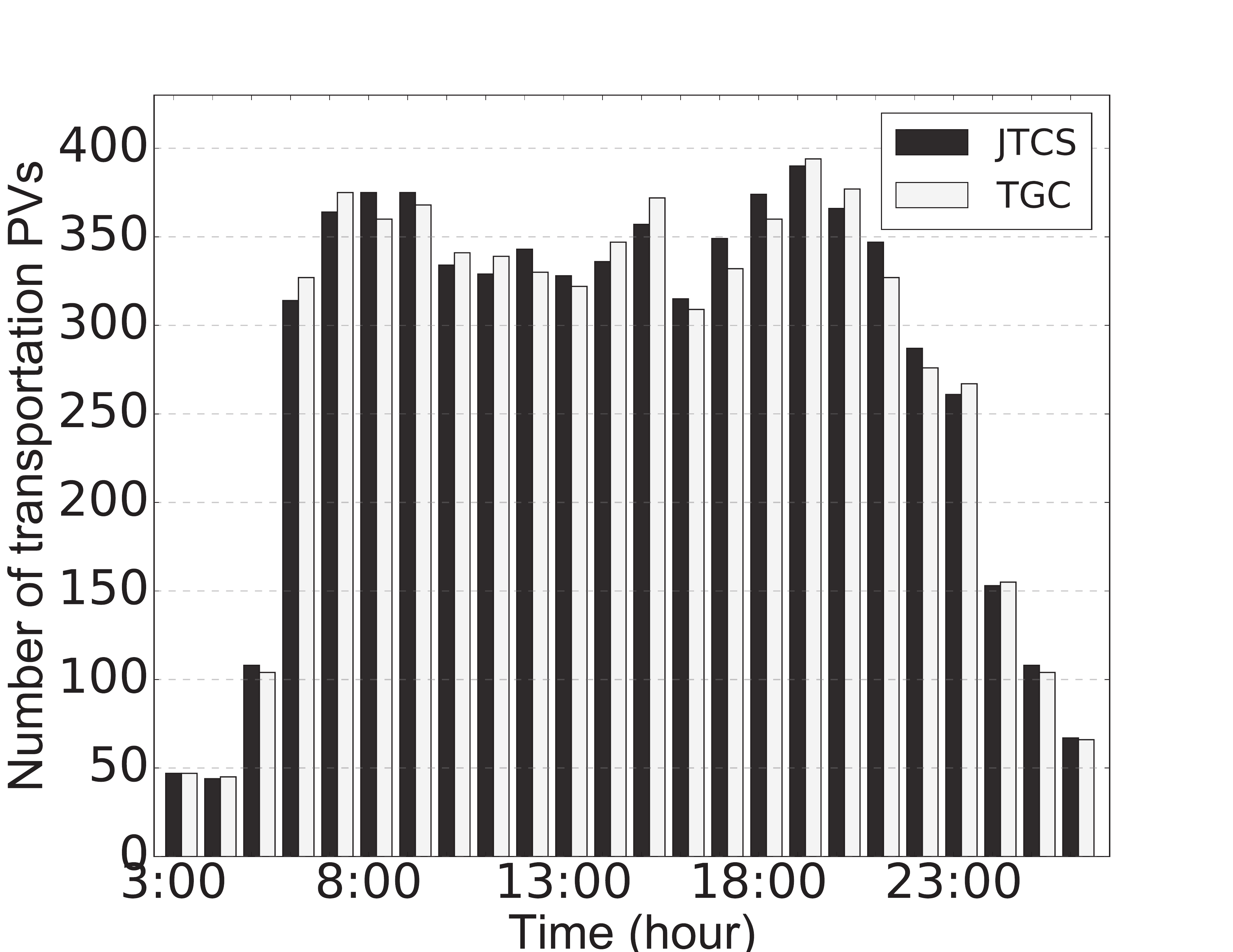}
  \caption{Number of transportation PVs in each hour.}
  \label{Fig:NumTransportationPVsEachHour}
  \end{minipage}
  \begin{minipage}[b]{0.33\linewidth}
  \centering
  \includegraphics[height=0.80\linewidth,width=0.99\linewidth]{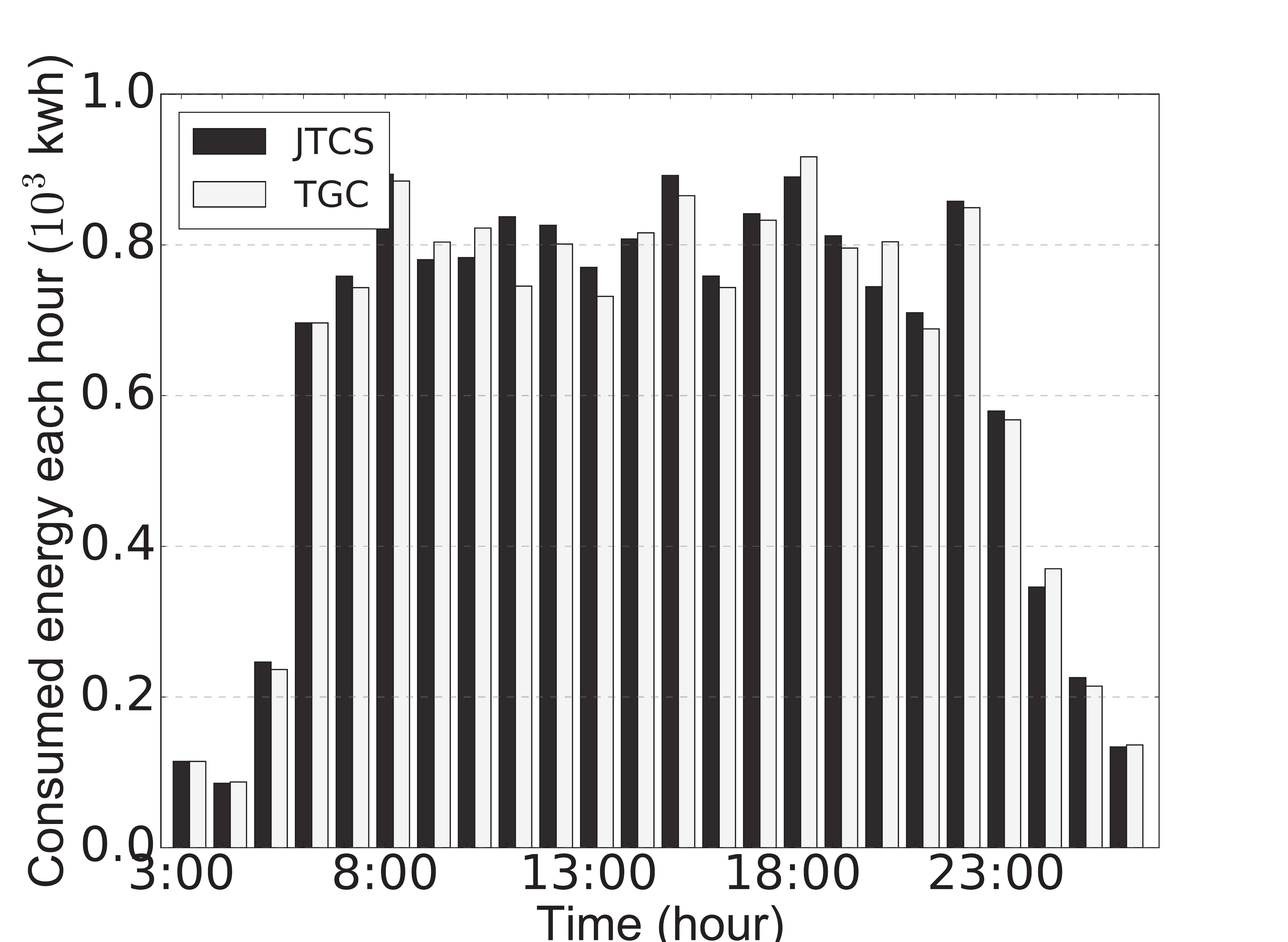}
  \caption{Consumed energy in each hour.}
  \label{Fig:ConsumedEnergyEachHour}
  \end{minipage}  \hspace{0.02in}
  \begin{minipage}[b]{0.33\linewidth}
  \centering
  \includegraphics[height=0.80\linewidth,width=0.99\linewidth]{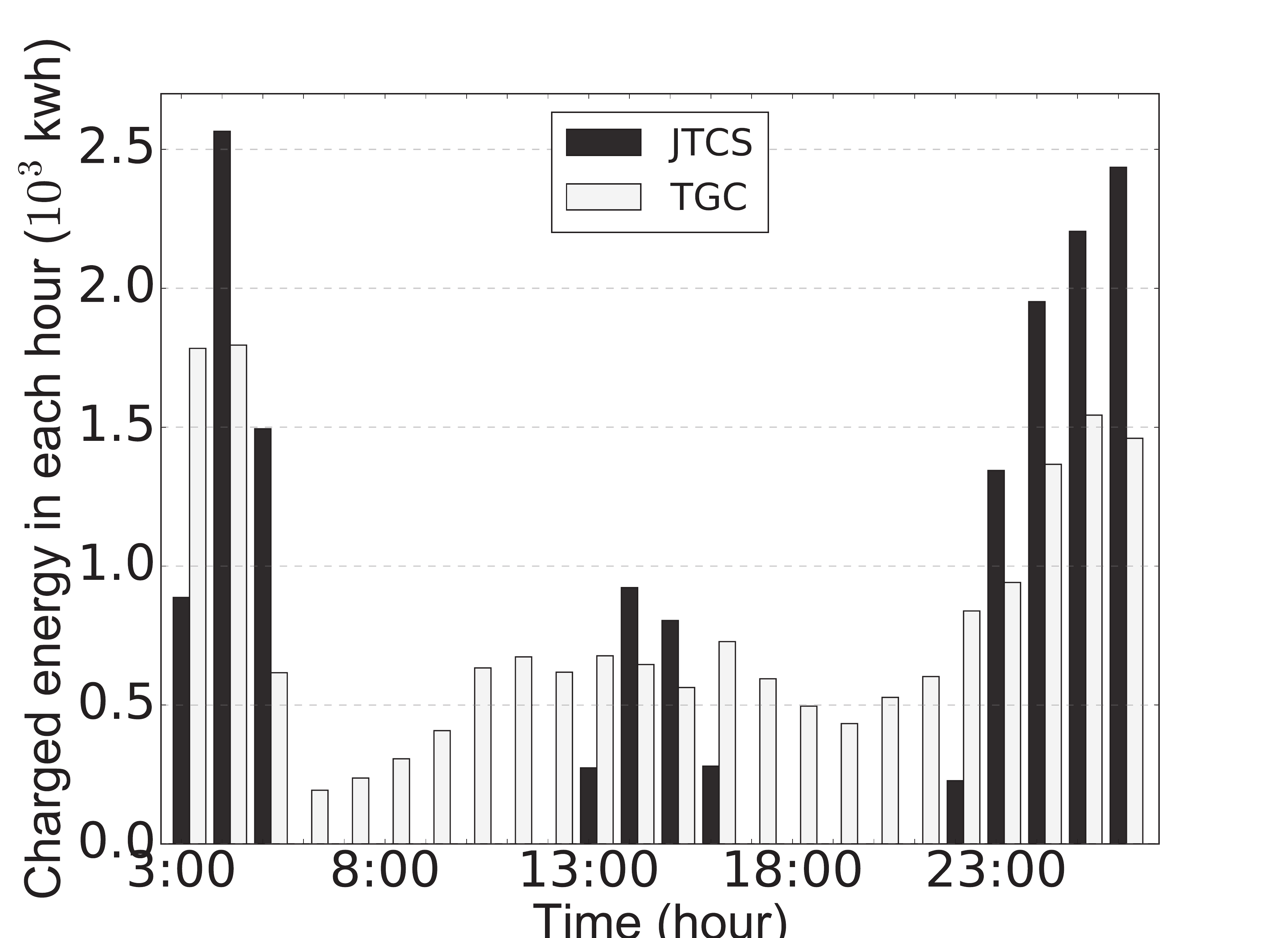}
  \caption{Charged energy in each hour.}
  \label{Fig:ChargedEnergyEachHour}
  \end{minipage}  \hspace{0.02in}
\end{figure*}
\begin{figure*}[htbp]
  \begin{minipage}[b]{0.33\linewidth}
  \centering
  \includegraphics[height=0.80\linewidth,width=0.99\linewidth]{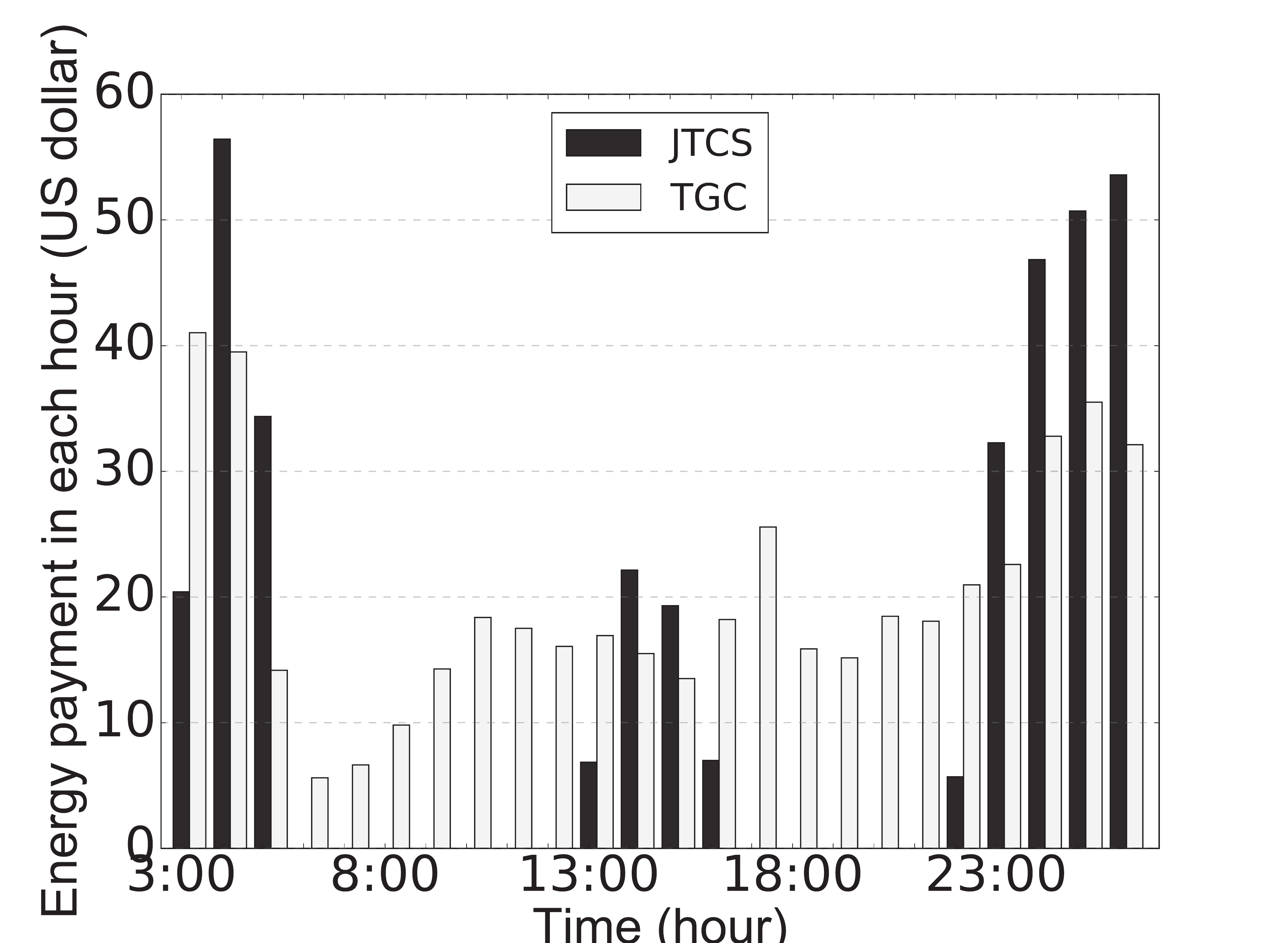}
  \caption{Energy payment in each hour.}
  \label{Fig:EnergyPaymentEachHour}
  \end{minipage}
  \begin{minipage}[b]{0.33\linewidth}
  \centering
  \includegraphics[height=0.80\linewidth,width=0.99\linewidth]{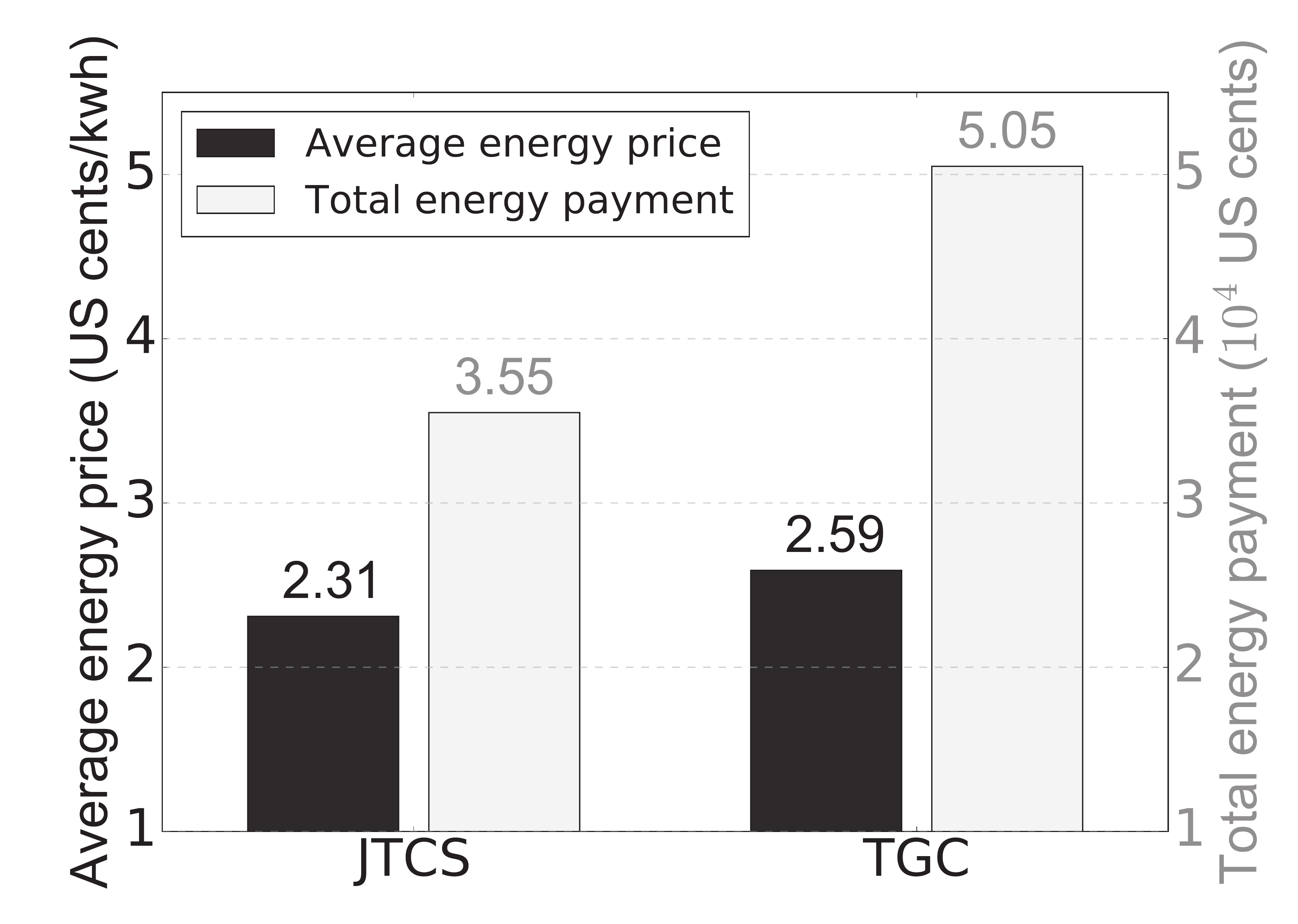}
  \caption{Average energy price and total payment.}
  \label{Fig:AverageEnergyPriceTotalPayment}
  \end{minipage}
  \begin{minipage}[b]{0.33\linewidth}
  \centering
  \includegraphics[height=0.80\linewidth,width=0.99\linewidth]{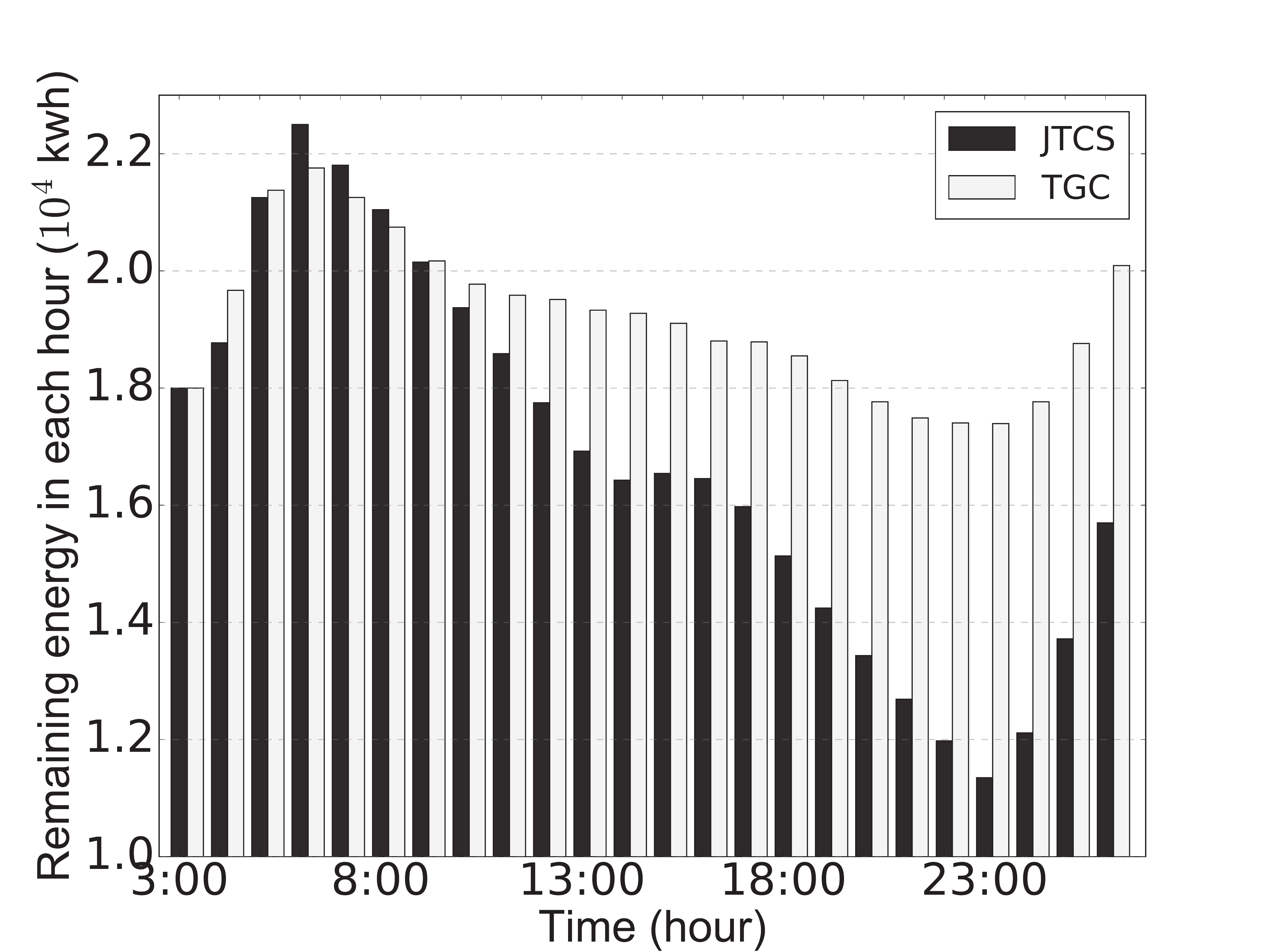}
  \caption{Remaining energy of all PVs.}
  \label{Fig:RemainingEnergyEachHour}
  \end{minipage}
\end{figure*}

The \emph{first} metric is the convergence of PVG strategies at the time 3:00, which is shown by Fig.~\ref{Fig:ConvergencePVG}. We see that, after seven iterations, the proposed JTCS algorithm with five PVGs converges to the NNE. Different PVGs have different best response strategies according to their transportation and charging demands.

The \emph{second} metric is the convergence of PVG utilities at the time 3:00, which is presented by Fig.~\ref{Fig:ConvergenceUtility}. We see that, all the utility values are negative since two items in the utility function (\ref{Eqn:u_i}) are negative and only one item is positive, and the utility values converge to the NNE with the converges of strategies.

If the number of PVs is set to 500, JTCS and TGC almost have the same transportation performance in terms of the number of transportation PVs in each hour, and the average trip time. In JTCS and TGC, the average trip time of passengers is about 14.6 minutes with the waiting time of 6.4 minutes and the travel time (from pickup time to dropoff time) of 8.2 minutes.

The \emph{third} metric is the number of transportation PVs in each hour in JTCS and TGC, which is depicted by Fig.~\ref{Fig:NumTransportationPVsEachHour}. We see that, the two schemes almost have the same performances, e.g., in transportation peak time (e.g., 8:00), about 400 PVs should serve passengers, i.e., at most 100 PVs can charge batteries. Generally, the number of PVs with sufficient energy in JTCS is smaller than that in TGC since all PVs except the transportation ones charge batteries in TGC, however, JTCS has almost the same transportation performance as TGC since it can coordinate the PVs in each region to satisfy the transportation demands of passengers, which can be seen from the first term in (\ref{Eqn:u_i}).

The \emph{fourth} metric is the consumed energy in each hour, which is shown by Fig.~\ref{Fig:ConsumedEnergyEachHour}. We see that, the energy consumptions in JTCS are similar to that in TGC, since the cloud in JTCS can coordinate the transportation demands of PVs in each region of the city, and considers about the energy consumptions in the next time slot, which is shown in the second constraint in Problem (\ref{Problem:PredictChargedEnergy}).

The \emph{fifth} metric is the charged energy in each hour, which is presented by Fig.~\ref{Fig:ChargedEnergyEachHour}. We see that, the charged energy during 3:00-3:59, 6:00-13:59 and 16:00-22:59 in TGC is obviously higher than that in JTCS. However, during 23:00-23:59, 1:00-2:59, and 4:00-5:59, the charged energy of JTCS is higher than that in TGC. We conclude that JTCS decides to charge more energy at low RTPs, however, to ensure the transportation performances of the next time slots, it decides to charge some energy even at high RTPs.

The \emph{sixth} metric is the energy payment in each hour using JTCS and TGC, which is depicted by Fig.~\ref{Fig:EnergyPaymentEachHour}. We see that, the profile of energy payment in each hour is similar to charged energy shown in Fig.~\ref{Fig:ChargedEnergyEachHour}. During night and noon time, the energy payment in JTCS is high, while at other time the energy payment may decrease to zero. The charged energy of JTCS and TGC is 15,392 kwh and 19,315 kwh respectively. We see that, PVs in JTCS charge less energy than that in TGC while this does not affect the transportation services.

The \emph{seventh} metric is the average energy price (the average energy payment per kwh, US cents/kwh) and total payment (US cents) using JTCS and TGC, which are together shown by Fig.~\ref{Fig:AverageEnergyPriceTotalPayment} with the same X-axis. We see that, the average energy price in the JTCS algorithm is 10.86\% less than TGC, and the total energy payment of JTCS and TGC is 355 and 505 US dollars respectively, reduced by 29.8\%. From this point of view, the average energy price in JTCS is reduced without reducing the transportation service quality since it considers both the transportation demands and charging demands of PVs. However, TGC only considers transportation with a greedy charging strategy, therefore, the charging costs are much higher than that in JTCS.

The \emph{eighth} metric is the remaining energy of all PVs in each hour, which is depicted by Fig.~\ref{Fig:RemainingEnergyEachHour}. We see that, the remaining energy in JTCS is generally less than that in TGC. The remaining energy in JTCS reduces quickly in transportation morning and evening peak time, e.g., 8:00 and 18:00, however, the remaining energy in GTC reduces very slowly, since more PVs decide to charge batteries even the RTPs are in high levels, which can reflect the effects of our scheme.

%Finally, we discuss the scalability of JTCS. We take the number of PVs as an example and there are two cases. 1) If the number of PVs is larger than 500, JTCS will have the same transportation performances as TGC, and better charging performances (denoted by the average energy prices, the less the better) than TGC. If more PVs are put, the transportation performances almost do not improve since in JTCS (or TGC) there are enough PVs with sufficient energy to serve passengers, while the charging performances in TGC will be dropped since more PVs will charge batteries with high prices in a greedy way, whereas, in fact, they do not need so many energy with such high prices to serve these trip requests. 2) If the number of PVs is less than 500, JTCS will have worse transportation performances than TGC since charging is more important than transportation, i.e., most of PVs try to charge more energy to satisfy charging demands and the more charged energy, the more transportation profits. In the real world, the latter case may not occur since PV systems have to consider the service guarantee for passengers, i.e., if too few PVs join the system, passengers can not choose PV systems for the bad service quality levels.

Finally, we discuss the scalability of JTCS over multiple days. In JTCS, we assume that the final remaining energy is not less than the initial energy in the day. To ensure the validness of JTCS over multiple days, we should guarantee the final energy of each day is not less than the consumed energy of the next day. For example, if we expect the final remaining energy of one day be more, we can revise the third constraint of (\ref{Problem:PredictChargedEnergy}) in \textbf{Algorithm 1} by increasing the right-hand bound, since the higher is the right-hand bound, the more final remaining energy PVs will have at the end of the day, and vice versa.

Our proposed approach can be easily extended to the case when the exact trip requests one day ahead are not known. In particular, the cloud can use the past requests of a similar day to predict the charging demands, e.g., on Monday, the cloud can use the requests of the previous Monday or the average of several Mondays. If the actual number of requests is significantly different from the past, the cloud can adjust the value of some parameters to adapt to the new scenarios. For example, if the number of requests becomes much larger than that of the past similar day, the cloud can increase the charging demands $E^{\text{+}}_t$ to ensure sufficient future energy.

\section{Conclusions} \label{Sec:Conclusion}

PV systems are new transportation systems for future smart cities, where PVs are typically self-driving electric vehicles. Transportation and charging coexist in PV systems, and to balance transportation and charging demands, we use a cake cutting game to capture the features of the PVGs and build utility models. Then we analyze the existence and uniqueness of Nash equilibrium in this game. Moreover, we propose the JTCS algorithm to achieve the unique normalized Nash equilibrium. Finally, we perform simulations based on the taxi trip data and smart grid data of New York to evaluate its performance. We find that, JTCS can provide almost the same transportation services as TGC, however, the average energy price is reduced by 10.86\% compared with TGC.

There are several future works on this research. As previously mentioned, the proposed JTCS algorithm assumes that the trip requests and the real-time price are known in advance, however, in the real world, it is impossible. Therefore, the real-time trip requests and real-time charging prices in uncertain traffic and smart grid settings should be considered in the future work. Obviously, new methods should be based on the predication of trip demands of passengers and the real-time electricity price. We only consider transportation and charging in PV systems, however, in fact, a small part of PVs may be willing to discharge batteries if the transportation demands are small in their regions. We will consider the interactions between transportation, charging, and discharging in the future works.

%\section*{References}

%\bibliography{Bib_bib}
%\bibliographystyle{ieeetr}

\vspace{-0.20in}
\begin{IEEEbiography}[{\includegraphics[width=1in,height=1.21in,clip,keepaspectratio]{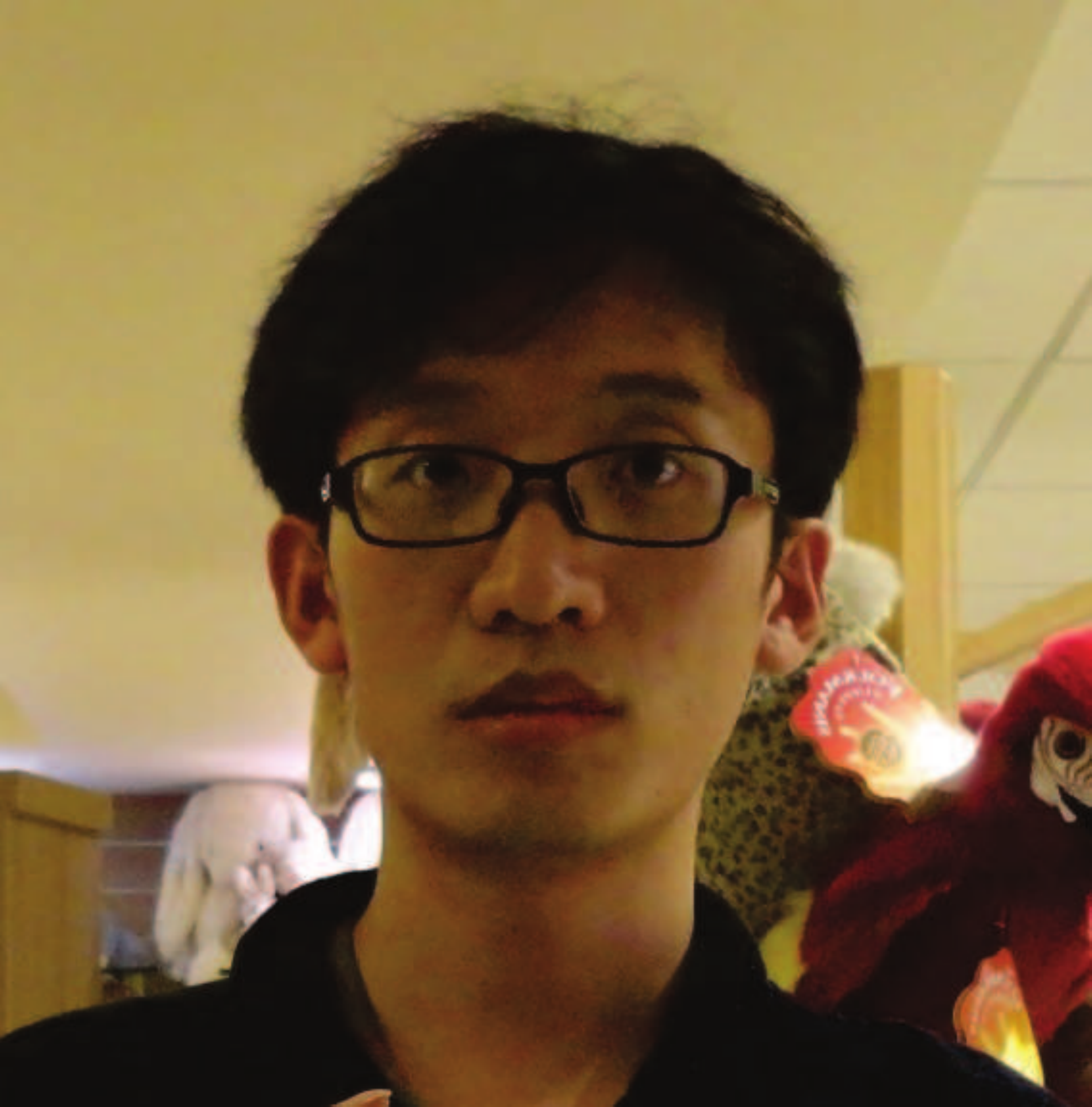}}]\\~~{\textbf{Ming~Zhu}} is now a postdoctoral researcher and assistant researcher in Shenzhen Institutes of Advanced Technology, Chinese Academy of Sciences, Shenzhen, China. He received the PhD degree in Computer Science and Engineering in Shanghai Jiao Tong University, Shanghai, China. A part of this work is finished in Shanghai Jiao Tong University.

His research interests are in the area of big data, artificial intelligence, internet of things, and wireless communications.
\end{IEEEbiography}

\vspace{-0.50in}
\begin{IEEEbiography}[{\includegraphics[width=1.40in,height=1.07in,clip,keepaspectratio]{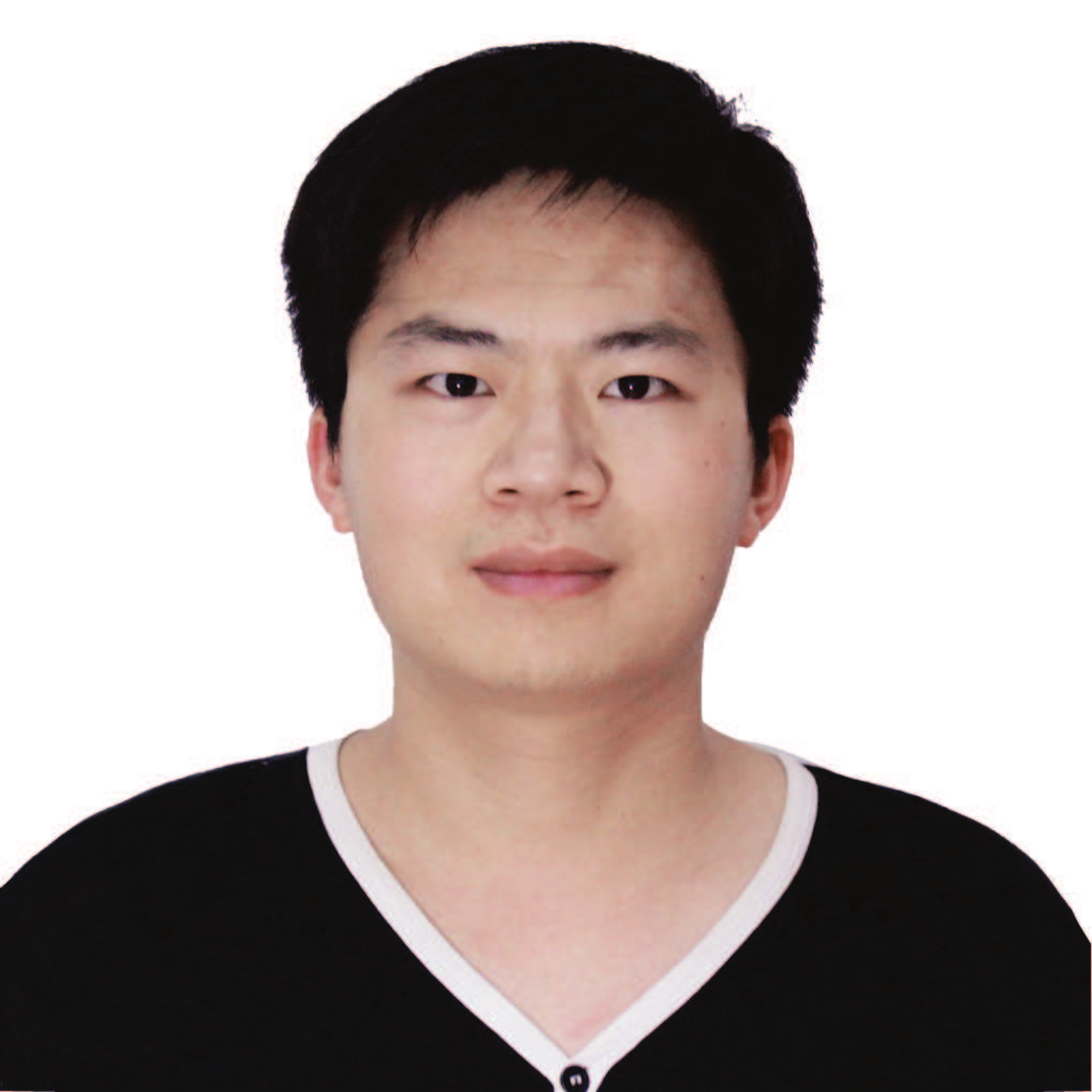}}]{Xiao-Yang~Liu} received his B.Eng. degree in computer science from Huazhong University of Science and Technology, China, in 2010. He is currently a joint PhD in the Department of Electrical Engineering, Columbia University, and in the Department of Computer Science and Engineer, Shanghai Jiao Tong University.

His research interests include tensor theory, deep learning, nonconvex optimization, big data analysis and homomorphic encryption, cyber-security and wireless communication.
\end{IEEEbiography}

\vspace{-0.50in}
\begin{IEEEbiography}[{\includegraphics[width=1in,height=1.21in,clip,keepaspectratio]{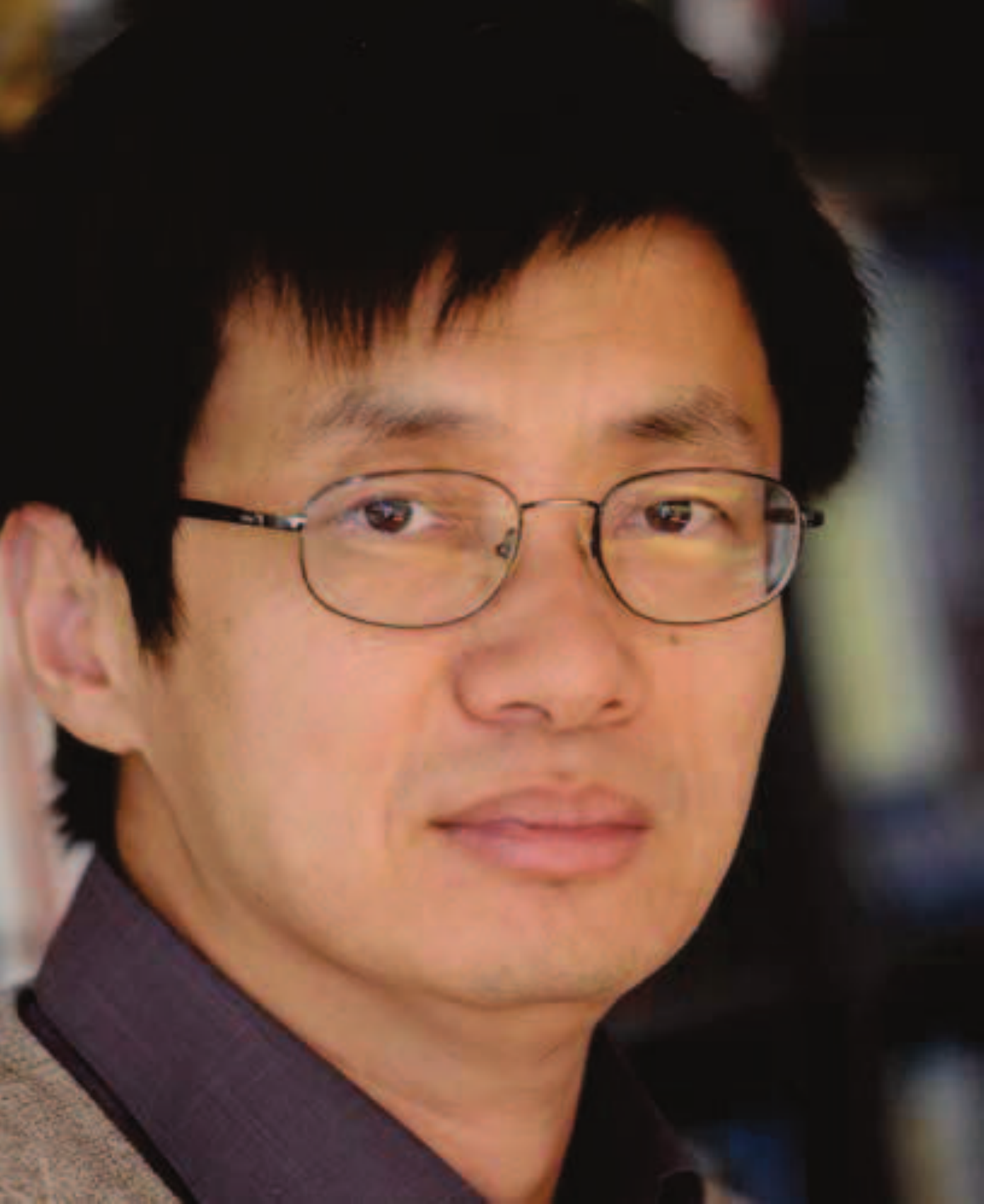}}]{Xiaodong~Wang} (S'98-M'98-SM'04-F'08) received the Ph.D. degree in electrical engineering from Princeton University. He is currently a Professor of electrical engineering with Columbia University, New York NY, USA. His research interests fall in the general areas of computing, signal processing, and communications. He has authored extensively in these areas. He has authored the book entitled Wireless Communication Systems: Advanced Techniques for Signal Reception, (Prentice Hall, 2003). His current research interests include wireless communications, statistical signal processing, and genomic signal processing. He has served as an Associate Editor of the IEEE TRANSACTIONS ON COMMUNICATIONS, the IEEE TRANSACTIONS ON WIRELESS COMMUNICATIONS, the IEEE TRANSACTIONS ON SIGNAL PROCESSING, and the IEEE TRANSACTIONS ON INFORMATION THEORY. He is an ISI Highly Cited Author. He received the 1999 NSF CAREER Award, the 2001 IEEE Communications Society and Information Theory Society Joint Paper Award, and the 2011 IEEE Communication Society Award for Outstanding Paper on New Communication Topics.
\end{IEEEbiography}

\end{document}